\documentclass[11pt]{article}

\usepackage{natbib}

\usepackage{amsmath}
\usepackage{amssymb}
\usepackage{amsthm}
\usepackage{bbm}
\usepackage{bm}

\usepackage{graphicx}

\usepackage{apalike}
\usepackage{caption}
\usepackage{subcaption}
\usepackage[left=25.4mm,right=25.4mm,top=38.1mm,bottom=38.1mm,a4paper]{geometry}
\usepackage{booktabs}

\usepackage{hyperref}
\newcommand{\footremember}[2]{
    \footnote{#2}
    \newcounter{#1}
    \setcounter{#1}{\value{footnote}}
}

\theoremstyle{definition}
\newtheorem{theorem}{Theorem}

\newtheorem{remark}[theorem]{Remark}
\newtheorem{assumption}{Assumption}

\newtheorem{example}{Example}

\newcommand{\E}{\mathbb E}
\newcommand{\e}{\mathrm e}
\newcommand{\Q}{\mathbb Q}
\newcommand{\D}{\mathrm{d}}
\newcommand{\F}{\mathcal F}

\begin{document}

\title{Analytic formula for option margin with liquidity costs under dynamic delta hedging\footremember{A}{This work was supported by “Human Resources Program in Energy Technology” of the Korea Institute of Energy Technology Evaluation and Planning (KETEP), granted financial resource from the Ministry of Trade, Industry \& Energy, Republic of Korea. (No. 20184010201680). Kyungsub Lee was supported by the 2020 Yeungnam University Research Grant.}}
\author{
  Kyungsub Lee\footremember{KLee}{ksublee@yu.ac.kr, Department of Statistics, Yeungnam University, Gyeongsan, Gyeongbuk 38541, Korea}
  \and Byoung Ki Seo\footremember{BKSeo}{(corresponding author) bkseo@unist.ac.kr, School of Business Administration, Ulsan National Institute of Science and Technology, Ulsan 44919, Korea, Tel.: +82 52 217 3150.}
  }
\date{}

\maketitle

\begin{abstract}
This study derives the expected liquidity cost when performing the delta hedging process of a European option.
This cost is represented by an integration formula that includes European option prices and a certain function depending on the delta process.
We first define a unit liquidity cost and then show that the liquidity cost is a multiplication of the unit liquidity cost, stock price, supply curve parameter, and the square of the number of options.
Using this formula, the expected liquidity cost before hedging can be calculated much faster than when using a Monte Carlo simulation.
Numerically computed distributions of liquidity costs in special cases are also provided.
\end{abstract}


\section{Introduction}

When an option trader has options, s/he may perform dynamic delta hedge to protect the price risk of the underlying. If the realized volatility is greater than the volatility s/he expected (i.e., implied volatility s/he entered into), then the option buyer[seller] will make a profit[loss]; Hence, option traders are often called volatility traders. It has been established that the instantaneous profit or loss from an option, and the respective dynamic delta hedge, is given by the following equation:

\begin{equation}
 \D\Pi_t = \dfrac 12 (\sigma_r^2 - \sigma_i^2) S_t^2 \Gamma_i \D t
 \end{equation}
 
where $\sigma_r$, $\sigma_i$, $S_t$, and $\Gamma_i$ is the realized volatility, the implied volatility, the underlying price, and the gamma, respectively. This means that if a volatility trader guesses the implied volatility with relative accuracy (i.e., similarly to the (to-be) realized volatility), then the option premium s/he paid[got] would be very close to the profit or loss of the underlying dynamic delta hedge.

However, this equilibrium is only possible based on the assumptions that there is no transaction cost, and that the trader can buy/sell the underlying at the price that is taken into consideration when the realized volatility is calculated (perhaps mid-price). Therefore, transaction cost and liquidity cost should be counted as margins when a volatility trader considers the price of an option. As such, the expected liquidity cost for dynamic delta hedging may also be regarded as a major factor determining the bid/ask spread of option prices.

In stock markets, traders buy or sell stocks through orders, including limit and market orders.
A limit order is an order to buy or sell a stock at a specified (or a better) price that remains on an order book (i.e., the list of all buy and sell limit orders), until it is filled, canceled, or given a designated time.
With limit orders, there is no price risk because the order is filled at the limit (or a better) price; however, in this case, the risk is associated with the time until, and uncertainty of, execution, as the limit orders do not guarantee execution.
A market order is an order to buy or sell a contract at the best current price, which is determined by the outstanding limit orders.
With market orders, one can buy or sell immediately, but there is a price risk depending on the number of contractible shares and the liquidity of the market (i.e., the status of existing limit orders).

The distribution of the limit prices and corresponding order sizes determine the cost of trading, especially when trading is based on market orders.
If there are fewer liquid limit orders, then there are larger bid-ask spreads, and trading is thus likely to be more expensive.
The liquidity cost during trading also depends on the shape of the limit order curve, which is determined by the outstanding limit price and the size of orders. This paper provides a semi-analytic formula for the expected liquidity cost of performing a delta hedging process on a European option within a presumed limit order curve.

A growing body of literature has focused on modeling and examining the dynamics of limit orders; however, this paper cites only a few recent studies.
\cite{Lo2002} proposed an econometric model of the time-to-execution of limit orders based on survival analysis; the findings show that the generalized gamma model fits historical data better than the theoretical first passage time approach.
\cite{Smith2003}, \cite{Cont2010} and \cite{Abergel} developed a stochastic model for the dynamics of the order flows, such as the limit orders, market orders, and cancellations based on the Poisson arrivals, without specific assumptions on the behaviors or preferences of the market participants.
For more information on the statistical property and modeling of the order book dynamics or the order-driven market, please see \cite{Maslov2000}, \cite{Bouchaud2002}, \cite{Hollifield2004}, \cite{Large2007}, \cite{Toke2011}, \cite{Gould2012}, \cite{Huang2012}, \cite{Malo}, \cite{Cont2013}, \cite{Xu}, and \cite{Chiu}.

Despite successfully incorporating the statistical properties of the order dynamics into a stochastic order flow model, and for the analytical purpose of the expected liquidity cost, this study employed a deterministic supply curve of the underlying stock in accordance with \cite{Cetin2004}, \cite{Cetin2006}, and \cite{Jarrow}.
The supply curve is the stock price per share that an investor pays or receives by a market order and is represented as a function of the quantity of stock executed in a market order.
In the literature, the authors examined the pricing of options in the extended framework of \cite{Black1973}, in terms of arbitrage pricing under illiquid conditions of an underlying asset.

There is numerous literature on computing the additional cost of delta-hedging under the relaxation of frictionless and competitive hypotheses in the Black-Scholes framework.
It is dated back to \cite{Leland}, who developed a modified option replicating strategy in the presence of transaction costs using a discrete-time replication scheme under the continuous-time framework.
\cite{Boyle} derived replicating strategies for European options with transaction costs in a binomial framework of \cite{cox1979option}.
With proportional transaction cost, those work provided the European option prices regardless of the supply of underlying asset, i.e., without the concept of liquidity risk due to the timing and size of trading.
Meanwhile, in the context of \cite{Cetin2004}, liquidity costs depend on the shape of outstanding limit orders of the time when a trade occurs, and hence are different from deterministic and proportional transaction costs or fees.
The supply curve of underlying is represented as a function that satisfies certain conditions such as twice continuously differentiability and non-decreasing in the trade size.
We argue that a linear approximation is sufficient to calculate the expected liquidity cost in both continuous and discrete trading.
This simplification leads that our result is also related to the classical work with proportional transaction cost by \cite{Leland} and \cite{Boyle}.

This study shows that the expected liquidity cost of the delta hedging process depends on the linear approximation of the supply curve, in both continuous and discrete trading.
Take, as an example, a European option with a non-perfectly liquid underlying stock in an order-driven stock market, where the exchange is executed using limit and market orders.
Furthermore, consider that the investor who sells the European option wishes to hedge the price risk by purchasing the delta amounts of the underlying stock.
In this case, the delta is the sensitivity of the option price with respect to the underlying stock price change, and the expected liquidity cost of the delta hedging is represented as an integration formula.
The total expected liquidity cost of the delta hedging process is represented by a multiplication of the unit liquidity cost, current stock price, supply curve parameter, and the square of the numbers of European options.
The unit liquidity cost will then be the liquidity cost of the delta hedging process and the unit supply curve parameter.

The integration formula for the expected liquidity cost includes European option prices with various strike prices and maturities and depends on the delta function of the option being hedged.
The methods calculated by the options' values in a static portfolio is an extension of previous work such as, \cite{Demeterfi1999}, \cite{Britten2000}, \cite{CarrMadan}, \cite{CarrWu}, \cite{Carr2013}, \cite{ChoeLee}, and \cite{Lee}.
The numerical computation of the integration is much faster than the computation of the liquidity cost based on a simulation.

The remainder of the paper is organized as follows.
Section~\ref{Sect:conti} derives the integration formula for the expected liquidity cost of delta hedging processes under a continuous trading assumption and provides an example in the Black-Scholes framework. 
The derivation is extended to the discrete trading case in Section~\ref{Sect:disc}.
Section~\ref{Sect:numeric} reports the result of numerical and simulation tests for the expected liquidity cost, and 
Section~\ref{Sect:dist} shows the distribution of the liquidity cost in special cases.
Section~\ref{Sect:LOB} explains the relationship between an existing queuing limit order model of \cite{Cont2010} and our approach to calculate the expected liquidity cost.
Finally, section~\ref{Sect:concl} concludes the paper.

\section{Expected liquidity cost under continuous trading}\label{Sect:conti}
This section introduces a probability space with a time index set, $[0,T]$, for some fixed $T>0$.
Let $(\Omega, \mathcal F, \mathbb P)$ be a filtered probability space with a filtration, $\{\mathcal F_t\}_{t\in [0, T]}$, where $\mathcal F_{T} = \mathcal F$.
This space satisfies the usual conditions.
The measure, $\mathbb P$, is the physical probability measure and there exists an equivalent martingale measure, $\Q$, under which all financial assets are priced (i.e., the discounted asset prices are martingale).
It is worth noting that all the processes introduced in this paper are defined in the probability space and are adapted to filtration.

If an investor wishes to perform a delta hedge process of an option with an underlying stock that is not perfectly liquid, they can immediately buy or sell the necessary number of shares for the process at the best available price (i.e., a market order) whenever they want.
The best price is determined by the outstanding limit orders, price specifications, and the number of shares for buying or selling, which is given by liquidity providers.

Table~\ref{Table:Limit} and the blue solid lines of Figure~\ref{Fig:Limit} illustrate an example of outstanding limit orders of Microsoft (left) and Intel (right) observed on 2012.06.21.
When the quantity (x-axis) of the graph is negative it is for a bid order, and when the quantity is positive, it is for an ask order.
The table shows that if the investor wants to buy 150,000 shares of Microsoft by market orders, they pay 30.14 dollars for 28,632 shares, 30.15 dollars for 83.663 shares and 30.16 dollars for 37,705 shares.
Let $m(q)$ be the limit order curve presented in the figure (blue solid lines) as a function of the quantity $q$.
This function is also called the marginal price curve.

The supply curve of the underlying stock at time $t$ is defined as
\begin{equation}
S(t,z) = \frac{1}{z}\int_0^z m(q) \D q \label{Eq:supply_curve}
\end{equation}
which represents the stock price that an investor pays (or receives) for order flow $z>0$ (or $z<0$) per share.
The terminology ``supply curve'' comes from \cite{Cetin2004}.
In some studies, the quantity, $\int_0^z m(q) \D q$, is called the total cost of a market order of $z$ shares of the stock \citep{Malo}.

When $z=0$, $S(t,z) = S_t$, which can be the mid-price of the best bid and ask prices.
In Figure~\ref{Fig:Limit}, the supply curves are indicated by the dotted red lines on the corresponding blue solid limit order curves. 
Practically, due to the minimum tick size, the plotted supply curve is a piecewise differentiable curve, and there is a jump at $z=0$.
On the other hand, for analysis of continuous trading, it is assumed that the supply curve is differentiable with respect to $z$, at least around $z=0$.
This theoretical assumption can be achieved by smoothing the curve around zero; an example of the smoothing procedure is provided later in the paper.
The analysis based on continuous trading and smooth supply curve is a theoretical issue.
From a practical perspective in the discrete trading case,
the differentiability of the supply curve is not required, as discussed in the next section.

According to \cite{Cetin2004}, the supply curve should satisfies:

\begin{enumerate}
	\item 
	$S(t, z)$ is $\mathcal F_t$-measurable and non-negative.
	\item
	$S(t, z)$ is non-decreasing in $z$, i.e., $z_1 \leq z_2$ implies $S(t, z_1) \leq S(t, z_2)$. 
	\item 
	$S(t, z)$ is $C^2$ in $z$.
	\item 
	$S(\cdot, 0)$ is a semi-martingale.
	\item 
	$S(\cdot, z)$ has continuous sample paths for all $z$.
\end{enumerate}
\cite{Cetin2006} assumed an exponential supply curve,
$$ S(t, z) = \e^{\alpha z} S(t, 0),$$
which is chosen for simplicity and the ease of generalization.
The other alternative forms are the supply curves with diminishing marginal price impact,
$ \alpha \mathrm{sign}(z)\sqrt{|z|}$ and $\alpha \mathrm{sign}(z)\log(1+|z|)$.	
The linear supply curve is the first-order approximation of the exponential form.
\cite{Blais2010} reported that linear supply curves are fitted well for highly liquid stocks
and jump linear curves are suitable for less liquid equities.



\begin{table}
\centering
\caption{Limit orders of Microsoft (MSFT) and Intel (INTC) at a certain moment on 2012.06.21.}\label{Table:Limit}
\begin{tabular}{rc|cr} 
\hline
Size & Bid price & Ask price & Size \\
\hline
51,326  &  30.13 &  30.14  & 28,632  \\
84,106  &  30.12 &  30.15  & 83,663  \\
8,706   &  30.11 &  30.16  & 66,999  \\
44,038  &  30.10 &  30.17  & 86,886  \\
167,571 &  30.09 &  30.18  & 110,006 \\
14,134  &  30.08 &  30.19  & 30,006  \\
46,380  &  30.07 &  30.20  & 72,106  \\
23,774  &  30.06 &  30.21  & 56,500  \\
23,646  &  30.05 &  30.22  & 36,532  \\
36,675  &  30.04 &  30.23  & 31,600  \\
\hline
\end{tabular} \,\,
\begin{tabular}{rc|cr} 
\hline
Size & Bid price & Ask price & Size \\
\hline
9,091   &  26.71 &  26.72  & 125,104 \\
34,683  &  26.70 &  26.73  & 174,683 \\
30,295  &  26.69 &  26.74  & 110,674 \\
47,583  &  26.68 &  26.75  & 59,778  \\
58,874  &  26.67 &  26.76  & 60,883  \\
79,774  &  26.66 &  26.77  & 39,946  \\
78,200  &  26.65 &  26.78  & 62,840  \\
11,200  &  26.64 &  26.79  & 28,655  \\
24,200  &  26.63 &  26.80  & 43,600  \\
52,000  &  26.62 &  26.81  & 47,600  \\
\hline
\end{tabular}

\end{table}

\begin{figure}
\centering
\includegraphics[width=0.45\textwidth]{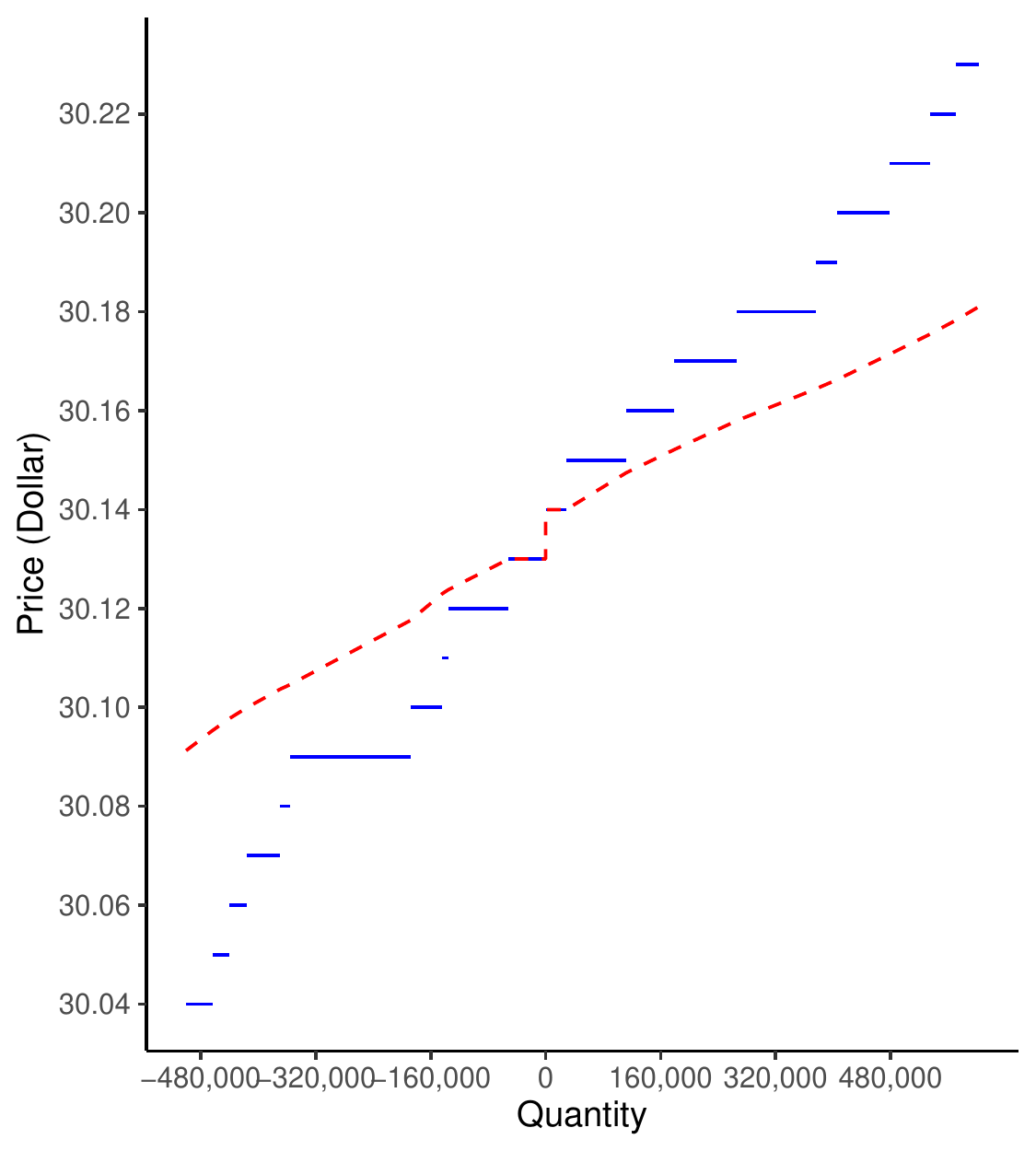} \,\,
\includegraphics[width=0.45\textwidth]{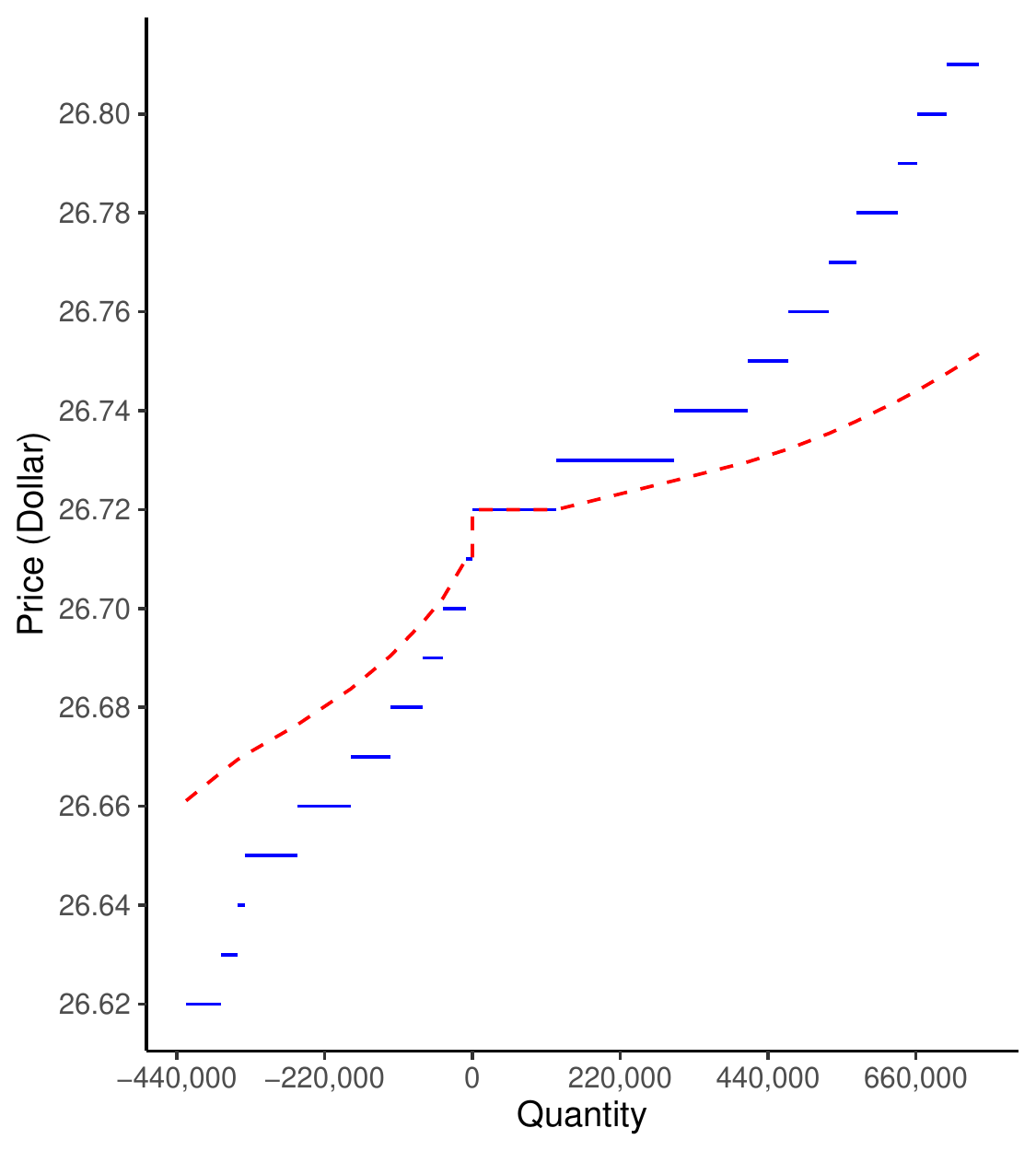}
\caption{Limit order curve (solid blue line) and corresponding supply curve (dotted red line) of Microsoft (left) and Intel (right) at a certain moment on 2012.06.21.}\label{Fig:Limit}
\end{figure}

In some studies, including \cite{Cetin2004}, \cite{Cetin2006} and \cite{Jarrow}, the liquidity cost during a trading period of $0<u\leq t$ is defined as
\begin{equation}
L_t = \sum_{0 <  u \leq t}\Delta Z_u [S(u,\Delta Z_u)-S(u,0)] + \int_0^t \frac{\partial S}{\partial z}(u,0)\D [Z]_u^c \label{Eq:LC}
\end{equation}
where $Z_t$ is the number of holding shares of the stock at $t$,
and $[Z]^c$ is the quadratic variation process of the continuous part of $Z$.
The first term on the right-hand side represents the liquidity cost for discrete trading.
The term $[S(u,\Delta Z_u)-S(u,0)]$ represents the additional cost for $\Delta Z_u$ shares of the stock due to illiquidity.
The second term, involving an integration with respect to the quadratic variation process of the continuous part of $Z$, represents the liquidity cost for continuous trading.
Roughly, the integration form in the equation originates from
\begin{equation}
\sum_{0<u\leq t} \frac{S(u,\Delta Z_u)-S(u,0)}{\Delta Z_u} (\Delta Z_u)^2 \rightarrow \int_0^t \frac{\partial S}{\partial z}(u,0)\D [Z]_u^c.
\end{equation}
In Eq.~\eqref{Eq:LC}, 
the trading cost at time 0 is executed,
because in a delta hedging process of a European option, the initial cost is simply calculated by the initial delta of said option.

As only the first-order partial derivative of $S(u, z)$ with respect to $z$ appears in the integration form in Eq.~\eqref{Eq:LC},
only the linear approximation of the supply curve at $z=0$ needs to be considered for analysis in the case of continuous trading. 
The linear approximation is represented by $S(t,z) \approx S_t + \alpha S_t z$ with some constant $\alpha$.

Consider the liquidity cost due to continuous trading with a continuous delta process $D(t, x)$ with $x=S_t$ at time $t$,
where the delta process is assumed to be a function of $t$ and $S_t$, as follows:
\begin{equation}
\int_0^T \frac{\partial S}{\partial z}(t,0)\D [Z]_t^c \approx \int_0^T \alpha S_t \D [D]_t = \int_0^T \alpha S_t \left(\frac{\partial D(t,S_t)}{\partial x} \right)^2 \D [S]_t \label{Eq:CLC}
\end{equation}
where $T$ is the maturity of the option and we use the fact that
$$ \D D_t =  \left(\frac{\partial D(t,S_t)}{\partial t} + \frac{1}{2} \frac{\partial^2 D(t,S_t)}{\partial x^2} \right) \D t + \frac{\partial D(t,S_t)}{\partial x}\D S_t $$
and
$$ \D [D]_t =  \left(\frac{\partial D(t,S_t)}{\partial x} \right)^2 \D [S]_t.$$
The superscript $c$ of $[S]$ is omitted because the continuous underlying price process is considered.
For instance, under the Black-Scholes framework, $\frac{\partial D(t,S_t)}{\partial x}$ is simply the Black-Scholes gamma and $\D [S]_t = \sigma^2 S_t^2 \D t$ with a volatility parameter $\sigma$.

The aim is to represent the risk-neutral expectation of the above equation in terms of European option prices with various maturities and strike prices.
A future price process with maturity $t$, defined by $F^{(t)}_u = \E^{\Q}[S_t | \F_u]$, for $u\leq t$, a martingale under $\Q$ is introduced.
For simplicity, assume that the instantaneous interest rate is constant at $r$.
The future price process can then be represented simply as $F^{(t)}_u = \e^{r(t-u)} S_u$.

Hereafter, $\E^{\Q}[\int_0^T \alpha S_t \D [D]_t]$ is derived from the European option prices.
If $F^{(t)}$ is continuous, then for a continuous function $g(x)$ and its antiderivative $G(x)$, it is known that
\begin{align}
\begin{aligned}
\int_0^t g\left(F^{(t)}_s \right) \D[F^{(t)}]_s &= 2 \left\{ \int_0^t \left(G \left(F^{(t)}_0\right) - G\left(F^{(t)}_s\right)\right) \D F^{(t)}_s  \right. \\
& \left. + \int_0^{F^{(t)}_0} g(k) \left( k-F^{(t)}_t \right)^+ \D k + \int_{F^{(t)}_0}^{\infty} g(k) \left( F^{(t)}_t-k \right)^+  \D k\right\}.\label{Eq:qintegration}
\end{aligned}
\end{align}
Please see \cite{ChoeLee} and \cite{Lee}.
Note that $F^{(t)}_t$ simply equals $S_t$.
This is derived from a comparison of the following two equations, involving It\`{o}'s formula for $\bar G$, which is the second antiderivative of $g$
\begin{align}
\bar G\left(F_t^{(t)}\right) - \bar G\left(F_0^{(t)}\right) = \int_0^t G\left(F_s^{(t)}\right) \D F_s^{(t)} + \frac{1}{2}\int_0^t g\left( F_0^{(t)}\right) \D [F^{(t)}]_s
\end{align}
and Taylor's theorem with the integral form of the remainder term
\begin{align}
\begin{aligned}
\bar G\left(F_t^{(t)}\right) - \bar G\left(F_0^{(t)}\right) 
&= G\left(F_0^{(t)}\right)\left(F_t^{(t)} - F_0^{(t)}\right) \\
&+\int_0^{F^{(t)}_0} g(k) \left( k-F^{(t)}_t \right)^+ \D k+ \int_{F^{(t)}_0}^{\infty} g(k) \left( F^{(t)}_t-k \right)^+  \D k. \label{Eq:Taylor1}
\end{aligned}
\end{align}

In addition, by taking $\Q$-expectation on both sides of Eq.~\eqref{Eq:qintegration} and assuming the stochastic integration with respect to $F^{(t)}$ in the equation is a $\Q$-martingale, we have
\begin{equation}
\E^{\Q} \left[ \int_0^t g\left(F^{(t)}_s \right) \D[F^{(t)}]_s \right] = 2 \e^{rt}\int_0^{\infty} g(k) \phi^{(t)}(S_0,k) \D k \label{Eq:expectation1}
\end{equation}
or equivalently, under the constant interest rate assumption,
\begin{align}
\E^{\Q} \left[ \int_0^t \e^{2r(t-s)}g\left(e^{r(t-s)}S_s \right) \D[S]_s \right] = 2 \e^{rt}\int_0^{\infty} g(k) \phi^{(t)}(S_0,k) \D k \label{Eq:expectation2}
\end{align}
where
\begin{equation}
\phi^{(t)}\left( S_0, k \right)=\left\{
      \begin{array}{ll}
        p^{(t)}\left( S_0, k \right) = \e^{-rt}\E^{\Q} [ (k-S_t)^+], & 0 < k \leq \e^{rt}S_0, \\
        c^{(t)}\left( S_0, k \right) = \e^{-rt}\E^{\Q} [ (S_t-k)^+], & \e^{rt}S_0 < k < \infty.
      \end{array}
    \right.
\end{equation}
Note that $p^{(t)}\left( S_0, k \right)$ and $c^{(t)}\left( S_0, k \right)$ are the European put and call option prices, with maturity $t$ and strike price $k$, respectively.

In Eq.~\eqref{Eq:expectation2}, the function $g$ in the integrand depends solely on the stock price process $S$,
and we need an extended version of the above argument to a function of two variables of $t$ and $S_t$ to deal with the expected liquidity cost represented in Eq.~\eqref{Eq:CLC}.
The following technical assumption is for treating a non-differentiable point in European option payoffs.

\begin{assumption}\label{Assum:1}
Let $\psi(t,x) \in C^1$ in both $t$ and $x$.
The partial derivatives, $\frac{\partial \psi}{\partial x}$ and $\frac{\partial^2 \psi}{\partial t \partial x}$, are absolutely continuous with respect to $x$, for all $t$, such that
\begin{align}
&\frac{\partial \psi}{\partial x}(t, b) - \frac{\partial \psi}{\partial x}(t, a) = \int_a^b \frac{\partial^2 \psi}{\partial x^2}(t, x) \D x \\
&\frac{\partial^2 \psi}{\partial t \partial x}(t, b) - \frac{\partial^2 \psi}{\partial t \partial x}(t, a) = \int_a^b \frac{\partial^3 \psi}{\partial t \partial x^2}(t, x) \D x
\end{align}
where the second partial derivatives in the r.h.s. can be the generalized derivatives in the distributional sense.
Let $F^{(T)}$ be a continuous process and
assume that
\begin{align}
\int_0^t \frac{\partial \psi}{\partial x}\left(s, F^{(T)}_s \right)\D F_s^{(T)},\quad t\leq T
\end{align}
is a $\Q$-martingale with respect to $\F$.
For example, if 
\begin{equation}
\frac{\partial \psi}{\partial x}\left(t, F^{(T)}_t \right) \in L^2_{\Q, [F^{(T)}]}([0,T]\times\Omega)
\end{equation}
where $L^2_{\Q, [F^{(T)}]}([0,T]\times\Omega)$ is the space of the adapted process X, such that
\begin{equation}
\E^{\Q} \left[\int_0^T X_u^2 \D [F^{(T)}]_u  \right] < \infty 
\end{equation}
then the martingale property of the stochastic integral is guaranteed, see \cite{Kuo}.
In addition, for convenience, let 
\begin{equation}
f(t,x)= \frac{\partial^2 \psi}{\partial x^2}(t,x).
\end{equation}

\end{assumption}

\begin{theorem}\label{Thm:implied_liquidity}
Under Assumption~\ref{Assum:1},
The risk-neutral expectation of the integration of $f\left(t,F^{(T)}\right)$ with respect to the quadratic variation of the futures price process $[F^{(T)}]$ is represented by the weighted European option prices.
That is
\begin{align}
\begin{aligned}
\E^{\Q} &\left[ \int_0^T \e^{2r(T-t)}f\left(t, \e^{r(T-t)}S_t \right)\D [S]_t \right] \\
=& - 2\int_0^T \int_0^\infty \e^{r(2T-t)}\frac{\partial f}{\partial t}\left(t,\e^{r(T-t)}k\right)\phi^{(t)}(S_0, k) \D k \D t \\
&+ 2\e^{rT}\int_0^\infty f(T,k) \phi^{(T)}(S_0,k) \D k.
\end{aligned}
\end{align}
\end{theorem}

\begin{proof}
By Taylor's theorem
\begin{align}
\begin{aligned}
\psi\left(T, F^{(T)}_T\right) ={}& \psi\left(T, F^{(T)}_0\right) + \frac{\partial \psi}{\partial x}\left(T,F^{(T)}_0\right)\left(F^{(T)}_T-F^{(T)}_0\right) \\
&+ \int_0^{F^{(T)}_0} \frac{\partial^2 \psi}{\partial x^2}(T,k)\left(k-F^{(T)}_T\right)^+ \D k \\
&+ \int_{F^{(T)}_0}^{\infty}\frac{\partial^2 \psi}{\partial x^2}(T,k)\left(F^{(T)}_T-k\right)^+ \D k.\label{Eq:Taylor_psi}
\end{aligned}
\end{align} 
In addition, according to Meyer-It\'{o}'s formula \citep{Protter},
\begin{align}
\begin{aligned}
\psi\left(T,F^{(T)}_T\right) ={}& \psi\left(0,F^{(T)}_0\right) + \int_0^T \frac{\partial \psi}{\partial t}\left(t, F^{(T)}_t\right)\D t + \int_0^T \frac{\partial \psi}{\partial x}\left(t,F^{(T)}_t\right) \D F^{(T)}_t \\
&+ \frac{1}{2}\int_0^T \frac{\partial^2 \psi}{\partial x^2}\left(t,F^{(T)}_t\right)\D [F^{(T)}]_t.
\end{aligned}
\end{align}
By combining the above results, we have
\begin{align}
\begin{aligned}
\frac{1}{2}\int_0^T \frac{\partial^2 \psi}{\partial x^2}&\left(t,F^{(T)}_t\right)\D [F^{(T)}]_t = \int_0^T \left(\frac{\partial \psi}{\partial t}\left(t,F^{(T)}_0\right) - \frac{\partial \psi}{\partial t}\left(t,F^{(T)}_t\right) \right)\D t \\
&+ \int_0^T \left( \frac{\partial \psi}{\partial x}\left(T,F^{(T)}_0\right) - \frac{\partial \psi}{\partial x}\left(t,F^{(T)}_t\right) \right) \D F^{(T)}_t \\
&+ \int_0^{F^{(T)}_0} \frac{\partial^2 \psi}{\partial x^2}(T,k)\left(k-F^{(T)}_T\right)^+ \D k \\
&+ \int_{F^{(T)}_0}^{\infty}\frac{\partial^2 \psi}{\partial x^2}(T,k)\left(F^{(T)}_T-k\right)^+ \D k.\label{Eq:thm1}
\end{aligned}
\end{align}

To deal with the integration term with respect to $t$ of the r.h.s. of the above equation, recall Taylor's theorem
\begin{align}
\begin{aligned}
 \frac{\partial \psi}{\partial t}\left(t,F^{(T)}_0\right) -& \frac{\partial \psi}{\partial t}\left(t,F^{(T)}_t\right) = - \frac{\partial^2 \psi}{\partial t \partial x}\left(t, F_0^{(T)}\right)\left( F_t^{(T)} - F_0^{(T)} \right) \\
&- \int_0^{F^{(T)}_0} \frac{\partial^3 \psi}{\partial t \partial x^2}(t,k')\left(k'-F_t^{(T)}\right)^+\D k'  \\
&- \int_{F^{(T)}_0}^{\infty} \frac{\partial^3 \psi}{\partial t \partial x^2}(t,k')\left(F_t^{(T)}-k'\right)^+\D k'.
\end{aligned}
\end{align}
Putting $k = k'\e^{-r(T-t)}$ and taking the $\Q$-expectation on both sides, we obtain
\begin{align}
\begin{aligned}
\E^{\Q}&\left[ \frac{\partial \psi}{\partial t}\left(t,F^{(T)}_0\right) - \frac{\partial \psi}{\partial t}\left(t,F^{(T)}_t\right) \right] \\
&= - \e^{r(2T-t)}\int_0^\infty \frac{\partial f}{\partial t}\left(t,\e^{r(T-t)}k\right) \phi^{(t)}(S_0,k) \D k. \label{Eq:Expect_dt}
\end{aligned}
\end{align}

Therefore, by taking $\Q$-expectation on Eq.~\eqref{Eq:thm1} and noting that the stochastic integration with respect to $F_t^{(T)}$ (the second term of the r.h.s. of the equation) is a $\Q$-martingale, we have
\begin{align}
\begin{aligned}
\E^{\Q} &\left[ \int_0^T f\left(t,F^{(T)}_t\right)\D [F^{(T)}]_t \right] \\
=& - 2\int_0^T \int_0^\infty \e^{r(2T-t)}\frac{\partial f}{\partial t}\left(t,\e^{r(T-t)}k\right)\phi^{(t)}(S_0,k) \D k \D t \\
&+ 2\e^{rT}\int_0^\infty f(T,k) \phi^{(T)}(S_0,k) \D k.
\end{aligned}
\end{align}
Finally, $F_t^{(T)} = \e^{r(T-t)}S_t $ and $\D[F^{(T)}]_t = \e^{2r(T-t)}\D[S]_t$ is used.
\end{proof}

In the above theorem, the expected liquidity cost for a delta hedging process is represented by a Riemann integration formula.
By putting
\begin{equation}
f(t,x) = \e^{-3r(T-t)}x \left(\frac{\partial D}{\partial  x}\left(t,\e^{-r(T-t)}x\right) \right)^2, \label{Eq:f}
\end{equation}
and applying the above theorem, 
\begin{align}
\begin{aligned}
\E^{\Q} &\left[\int_0^T \alpha S_t \D [D]_t  \right] =\E^{\Q} \left[\int_0^T \alpha S_t \left(\frac{\partial D(t,S_t)}{\partial x} \right)^2 \D [S]_t \right] \\
={}&\alpha\E^{\Q} \left[ \int_0^T \e^{2r(T-t)}f\left(t, \e^{r(T-t)}S_t \right)\D [S]_t \right]\\
={}& - 2\alpha\int_0^T \int_0^\infty \e^{r(2T-t)}\frac{\partial f}{\partial t}\left(t,\e^{r(T-t)}k\right)\phi^{(t)}(S_0,k) \D k \D t \\
&+ 2\alpha\e^{rT}\int_0^\infty f(T,k) \phi^{(T)}(S_0,k) \D k. \label{Eq:BSELC}
\end{aligned}
\end{align}
The expected liquidity increases linearly with $\alpha$ and quadratically with the number of options being hedged.
The expected liquidity cost can be conveniently represented in the following way.

\begin{theorem}\label{Thm:unit}
Consider the delta hedging process of $N$ numbers of European options with a strike $K$, maturity $T$, and current underlying price $S_0$.
Let $M=K/S_0$ and define the unit liquidity cost by
\begin{align}
\begin{aligned}
I ={}& - 2\int_0^T \int_0^\infty \e^{r(2T-t)}\frac{\partial f}{\partial t}\left(t,\e^{r(T-t)}k;M,T\right)\phi^{(t)}(1,k) \D k \D t \\
&+ 2\e^{rT}\int_0^\infty f(T,k;M,T) \phi^{(T)}(1,k) \D k
\end{aligned}
\end{align}
where
\begin{equation}
f(t,x;M,T) = \e^{-3r(T-t)}x \left(\frac{\partial \bar D}{\partial  x}\left(t,\e^{-r(T-t)}x;M,T\right) \right)^2 
\label{Eq:f2}
\end{equation}
and $\bar D(t,x;M,T)$ is the delta function of the European option with a strike $M$, maturity $T$, and current underlying price 1;
in other words, the values are denominated by the current underlying stock price, $S_0$.
The expected liquidity cost of the delta hedging process of $D=N\bar D$ can then be expressed as 
\begin{equation}
\E^{\Q} \left[ \int_0^T \alpha S_t \D [D]_t \right]=\alpha N^2 S_0 I.
\end{equation}
\end{theorem}
\begin{proof}
Let $\bar S_t = S_t /S_0$.
Then
\begin{align}
\E^{\Q} \left[ \int_0^T \alpha S_t \D [D]_t \right] = \alpha N^2 S_0 \E^{\Q} \left[ \int_0^T \alpha \bar S_t \D [\bar D]_t \right]
\end{align}
and use Theorem~\ref{Thm:implied_liquidity}.
\end{proof}

\begin{example}
Suppose that the Black-Scholes delta function is used to hedge the option.
As an exemplary case, the trading strategy is the trader's own choice and does not need to be an optimal replication.
The delta process of the European call option with strike $K$ and maturity $T$ is
\begin{equation}
D(t,S_t) = \Phi(h_t), \quad h_t = \frac{\log\frac{S_t}{K} + \left(r + \frac{\sigma^2}{2}\right)(T-t)}{\sigma\sqrt{T-t}},
\end{equation}
for $0\leq t<T$, where $\Phi$ denotes the standard cumulative normal distribution.
Therefore, the gamma of the option $x=S_t$ is
\begin{equation}
\frac{\partial D(t,x)}{\partial x} = \frac{1}{x\sigma\sqrt{2\pi(T-t)}}\exp \left[ -\frac{1}{2}\left\{\frac{\log\frac{x}{K} + \left(r + \frac{\sigma^2}{2} \right) (T-t)}{\sigma\sqrt{T-t}} \right\}^2\right]
\end{equation}
and
\begin{equation}
f(t,x) = \frac{1}{2\pi x\sigma^2(T-t)} \exp\left[-r(T-t)- \left\{\frac{\log\frac{x}{K} + \frac{\sigma^2}{2}(T-t)}{\sigma\sqrt{T-t}} \right\}^2\right].
\end{equation}
and
\begin{align}
\begin{aligned}
\frac{\partial f}{\partial t}(t,x) ={}&\frac{\sigma^2 \left\{4+4 r (T-t)+\sigma^2 (T-t)\right\} (T-t)-4 \log^2\frac{x}{K}}{8 \pi \sigma^4 (T-t)^3 x}\\
&\times\exp \left[ -r(T-t) - \left\{ \frac{\log\frac{x}{K} + \frac{\sigma^2}{2}(T-t)}{\sigma\sqrt{T-t}} \right\}^2 \right].\label{Eq:BSft}
\end{aligned}
\end{align}
For the above formula to be valid, we assume that there are no additional costs such as transaction costs.
With an infinite variation hedging strategy, the theoretical super-replication price of an option is explained in 
\cite{Cetin2010}.
\end{example}

\begin{remark}
To represent the expected liquidity cost in terms of the European option prices, the $\Q$-expectation is taken as the liquidity cost in this study.
On the other hand, for example, under the Black-Scholes framework (i.e., if the underlying stock price process follows geometric Brownian motion), one can simply change $r$ as the drift of the stock price process in order to derive the expectation in terms of physical probability.
\end{remark}

\section{Expected liquidity cost under discrete trading}\label{Sect:disc}
The previous section considers the expected liquidity cost under the continuous trading assumption (with an infinite variation trading strategy).
However, the application of continuous trading is not practical here.
This section examines the expected liquidity cost under the discrete time trading of the delta hedging process.
The methodology is similar to that of the continuous case.

In Figure~\ref{Fig:DSC}, an example of the supply curve is provided and its linear approximations plotted.
The straight line represents the first-order approximation around zero, and was used to calculate the expected liquidity cost under continuous time trading in the previous section.
For discrete time trading, assume that the hedger rebalances the portfolio when the necessary change in delta exceeds a certain fixed quantity, for example, 3,000 shares.
According to the given supply curve in the figure, the first-order approximation of the supply curve for the discrete trading (the dashed line in the figure) is different from the one of continuous trading.
Therefore, the liquidity supply curve parameter, $\alpha'$, is used for discrete trading, which may differ from $\alpha$.

The slope of the supply curve is assumed to be symmetrical as \cite{Biais1995}, \cite{Bouchaud2002}, and \cite{Potters2003} reported no significant difference between the mean bid and ask side depths of the outstanding limit orders.
In addition, studies have also reported on the asymmetrical curve shape; see \cite{Gu2008} and \cite{Malo}.

\begin{figure}
\centering
\includegraphics[width=0.45\textwidth]{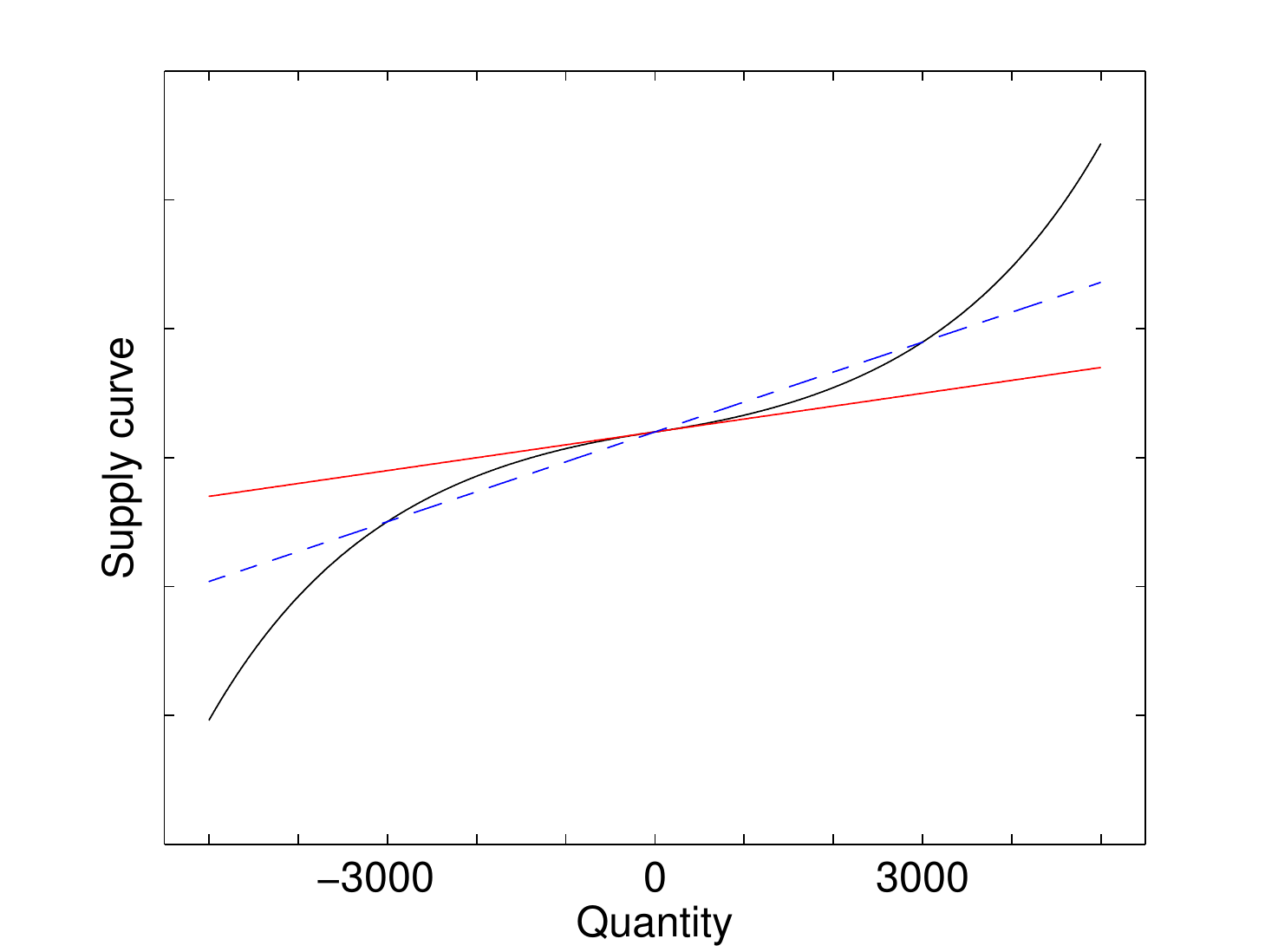}
\caption{Example of the supply curve and first-order approximation for continuous trading (straight line) with a slope $\alpha$, and for discrete trading (dashed line) with slope $\alpha’$}\label{Fig:DSC}
\end{figure}

In addition, it is assumed that the limit order book is perfectly resilient and that the supply curve is exogenous.
If the hedger buys or sells certain amounts of underlying stock to rebalance the delta, which results in a depletion of existing limit orders,
then the limit orders are replenished immediately so that there are no effects on the following trades or the other investors’ trades.

With a delta process, $D(t,S_t)$, the discrete trading liquidity cost term in Eq.~\eqref{Eq:LC} is rewritten as
\begin{equation}
\sum_{0 < u \leq t}\Delta D(u,S_{u}) [S(u,\Delta D(u,S_u))-S_{u}] = \sum_{0 < u \leq t}\alpha'S_u (\Delta D(u,S_u))^2
\end{equation}
where $\alpha' S_u \Delta D(u,S_u)$ is a linear form of $S(u,\Delta D(u,S_u))-S_{u}$.
Note that the above is not an approximation, but exactly equivalent under the discrete trading strategy with a fixed $\Delta D$ and supply curve.
If $\Delta D$ or the supply curve is not fixed, then the liquidity cost cannot be represented by a single parameter $\alpha'$.
In addition, by taking the first-order approximation of a delta function, the liquidity cost can be approximated as
\begin{equation}
\sum_{0 < u \leq t}\alpha'S_u \left(\frac{\partial D(u,S_u)}{\partial x}\right)^2(\Delta S_u)^2 \label{Eq:LC_disc}
\end{equation}
and our goal is to compute the expectation of the above equation.

To treat discrete trading, it is convenient to consider a pure jump process version of the underlying price and the corresponding futures prices (even though the actual price process is continuous) in the following way.
Let $0 = \tau_0 < \cdots < \tau_i < \tau_{i+1} < \cdots $ be the stopping time such that each $\tau_i$ is the rebalancing time of the delta hedging portfolio, by changing the delta with some fixed $\Delta D$.
In addition, let $\bar F^{(t)}$ be the pure jump version of the futures process $F^{(t)}$, such that
$ \bar F^{(t)}_{\tau_i} = F^{(t)}_{\tau_i} $ and $\bar F^{(t)}_u = \bar F^{(t)}_{\tau_i}$ for $ \tau_i < u < \tau_{i+1}$.
In particular, in $\bar F^{(T)}$, with the maturity $T$ of the option being hedged, let $\bar F^{(T)}_T = F^{(T)}_T$.

Suppose that $f(t,x)$ and $\psi(t,x)$ satisfy the assumption~\ref{Assum:1};
then, according to It\`{o}'s formula for the pure jump process, $\bar F^{(T)}$,
\begin{equation}
\psi\left(T,\bar F^{(T)}_T \right) = \psi\left(0,\bar F_0^{(T)}\right) + \int_0^T \frac{\partial \psi}{\partial t}\left(t, \bar F_t^{(T)}\right)\D t + \sum_{0<t\leq T} \Delta \psi\left(t, \bar F_t^{(T)}\right).
\end{equation}
Compared to the Taylor expansion in Eq.~\eqref{Eq:Taylor_psi}, we have
\begin{align}
\begin{aligned}
\sum_{0<t\leq T} \Delta \psi\left(t,\bar F_t^{(T)}\right) ={}& \int_0^T \left(\frac{\partial \psi}{\partial t}\left(t, \bar F^{(T)}_0\right) - \frac{\partial \psi}{\partial t}\left(t, \bar F^{(T)}_t\right) \right)\D t\\
&+ \int_0^{F_0} \frac{\partial^2 \psi}{\partial x^2}(T,k)\left(k- \bar F^{(T)}_T\right)^+ \D k \\
&+ \int_{F^{(T)}_0}^{\infty}\frac{\partial^2 \psi}{\partial x^2}(T,k)\left(\bar F^{(T)}_T-k\right)^+ \D k.
\end{aligned}
\end{align}

By Taylor’s expansion and Eq.~\eqref{Eq:Expect_dt}, 
\begin{align}
\begin{aligned}
&\E^{\Q}\left[ \sum_{0<t\leq T} \left\{ \frac{\partial \psi}{\partial x}\left(t,\bar F_{t-}^{(T)}\right)\Delta \bar F_t^{(T)} + \frac{1}{2}\frac{\partial^2 \psi}{\partial x^2}\left(t,\bar F_{t-}^{(T)}\right)\left(\Delta \bar F_t^{(T)}\right)^2 + \cdots \right\} \right]\label{Eq:expansion}\\
=&\E^{\Q}\left[ \sum_{0<t\leq T} \Delta \psi\left(t,\bar F_t^{(T)}\right) \right] \\
=& - \int_0^T \int_0^\infty \e^{r(2T-t)}\frac{\partial f}{\partial t}\left(t,\e^{r(T-t)}k\right)\phi^{(t)}(k) \D k \D t + \e^{rT}\int_0^\infty f(T,k) \phi^{(T)}(k) \D k.
\end{aligned}
\end{align}
Note that 
\begin{align}
\E^{\Q}\left[ \sum_{0<t\leq T} \frac{\partial \psi}{\partial x}\left(t, \bar F_{t-}^{(T)}\right)\Delta \bar F_t^{(T)}\right] = 0
\end{align}
Assuming that the higher-order terms of the jump in Eq.~\eqref{Eq:expansion} are negligible, 
the following approximation can be derived:
\begin{align}
\begin{aligned}
\E^{\Q}& \left[ \sum_{0<t\leq T}  \frac{\partial^2 \psi}{\partial x^2}\left(t, \bar F_{t-}^{(T)}\right)\left(\Delta \bar F_t^{(T)}\right)^2 \right] \\
\approx& - 2\int_0^T \int_0^\infty \e^{r(2T-t)}\frac{\partial f}{\partial t}\left(t,\e^{r(T-t)}k\right)\phi^{(t)}(k) \D k \D t \\
&+ 2\e^{rT}\int_0^\infty f(T,k) \phi^{(T)}(k) \D k.
\end{aligned}
\end{align}
or
\begin{align}
\begin{aligned}
\E^{\Q} &\left[ \sum_{0<t\leq T} \e^{2r(T-t)}f\left(t, \e^{r(T-t)}\bar S_{t-} \right) (\Delta \bar S_t)^2  \right] \\
\approx& - 2\int_0^T \int_0^\infty \e^{r(2T-t)}\frac{\partial f}{\partial t}\left(t,\e^{r(T-t)}k\right)\phi^{(t)}(k) \D k \D t\\
&+ 2\e^{rT}\int_0^\infty f(T,k) \phi^{(T)}(k) \D k.\label{Eq:D_ELC}
\end{aligned}
\end{align}
As in the continuous trading case, if
\begin{align}
f(t,x) =  \e^{-3r(T-t)}x \left(\frac{\partial D}{\partial  x}\left(t,\e^{-r(T-t)}x\right) \right)^2, 
\end{align}
is chosen,
then the expected liquidity cost can be approximated using the integration formula in Eq.~\eqref{Eq:D_ELC}.

We examine whether the higher terms in Eq.~\eqref{Eq:expansion} are negligible.
The third-order term is represented by
\begin{align}
\frac{1}{6}\frac{\partial^3 \psi}{\partial x^3}\left(t,\bar F_{t-}^{(T)}\right)\left(\Delta \bar F_t^{(T)}\right)^3 = \frac{1}{6} \left(\bar F_{t-}^{(T)}\right)^3\frac{\partial f}{\partial x}\left(t, \bar F_{t-}^{(T)}\right)\left(\frac{\Delta \bar F_t^{(T)}}{\bar F_{t-}^{(T)}}\right)^3.
\end{align}
Under the Black-Scholes framework, as in Eq.~\eqref{Eq:BSft}, $x^3 \frac{\partial f(t,x)}{\partial x}$ is continuous and converges to zero as $x$ goes to infinity for every $t$, and is hence bounded over $(0,\infty)$.
The same argument is applied to the other higher terms.
Therefore, the higher-order terms are represented by $O\left( \left(\Delta \bar F_t^{(T)}/\bar F_{t-}^{(T)}\right)^3 \right)$ and, in this study, it is assumed that the third moment of the return is negligible.
Figure~\ref{Fig:BSft} gives an example of $x^3 \frac{\partial f(t,x)}{\partial x}$ with 
$T = 0.1, t = 0, K = 1, r = 0.05, \sigma = 0.3, \alpha = 1$.

In summary, with discrete time trading of delta hedging, the expected liquidity cost can be calculated using the same method as in the continuous trading case, if the hedger rebalances the portfolio and the change in delta is equal to a fixed value.

\begin{figure}
\centering
\includegraphics[width=0.45\textwidth]{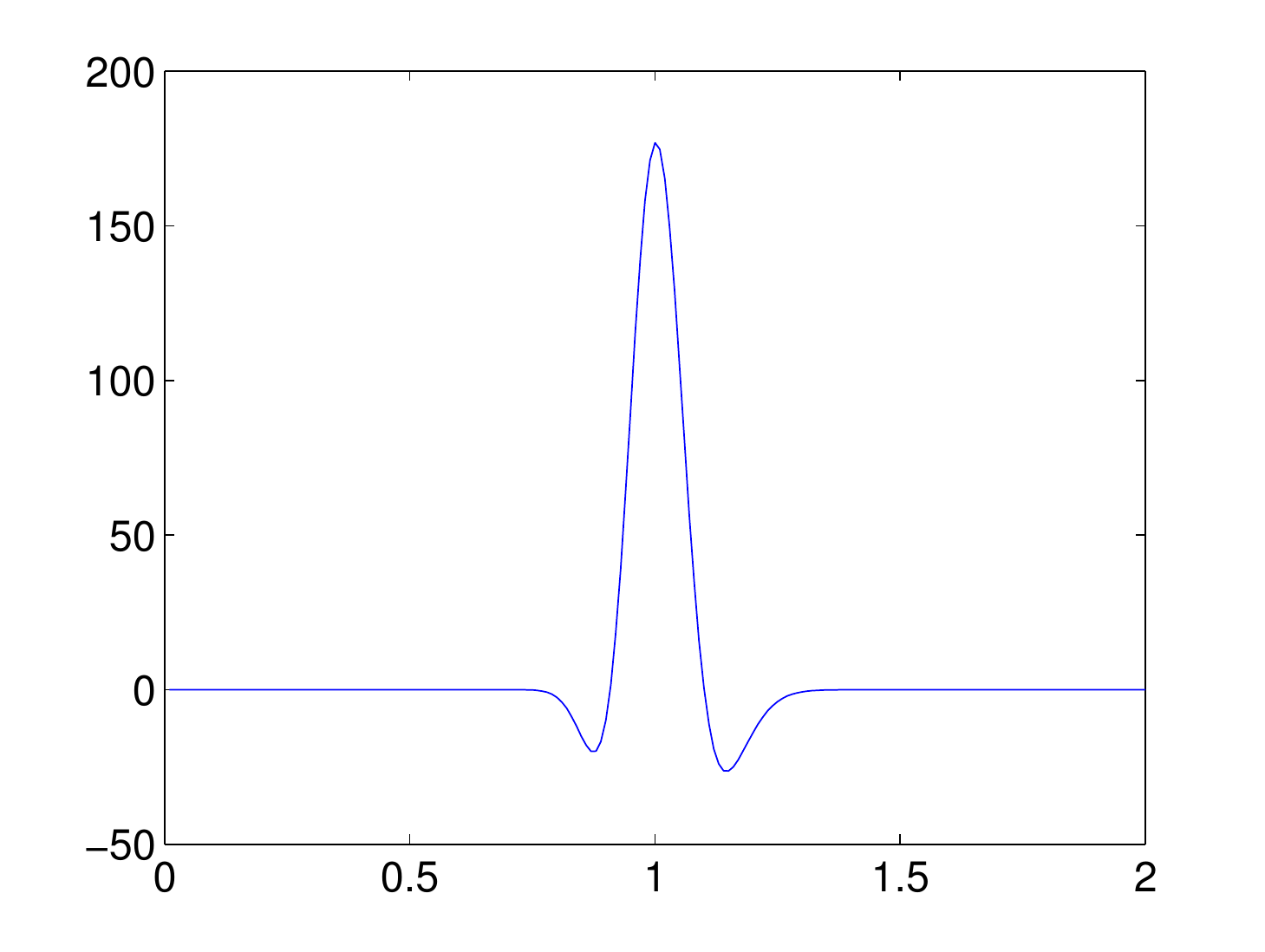}
\caption{Example of $x^3 \frac{\partial f(t,x)}{\partial x}$ with 
$T = 0.1, t = 0, K = 1, r = 0.05, \sigma = 0.3, \alpha = 1$}\label{Fig:BSft}
\end{figure}

\section{Numerical and simulation study}\label{Sect:numeric}

This section calculates the expected liquidity cost of the delta hedging process of European options using the above formulas and the Monte Carlo simulation.
The Black-Scholes delta and the Black-Scholes option prices are used for the hedging process and the weighted numerical integration, respectively.
According to Theorem~\ref{Thm:unit}, it is important to examine the property of the unit liquidity cost $I$ to calculate the expected liquidity cost of the delta hedging, because the cost is represented by $\alpha S_0 N^2 I$.

First, the numerical integration of $I$ was compared with the simulation results for the delta hedging processes of the European call options with maturity $T$ and strike price $K$.
This is because at time $t=T$, the delta function of a European call option is not continuous with respect to the underlying price;
hence, $f$ defined by Eq.~\eqref{Eq:f} or \eqref{Eq:f2} is also discontinuous.
For the numerical integration, the expected liquidity cost is approximated by taking $T'$ to be very close to $T$, where $f$ is continuous and differentiable with respect to $x$.
The approximation to the numerical integration of the unit liquidity cost is 
\begin{align}
\begin{aligned}
I \approx& - 2\int_0^{T'} \int_0^\infty \e^{r(2T' -t)}\frac{\partial f}{\partial t}\left(t,\e^{r(T'-t)}k;M,T\right)\phi^{(t)}(1,k) \D k \D t \\
&+ 2\e^{rT'}\int_0^\infty f(T',k;M,T) \phi^{(T')}(1,k) \D k.\label{Eq:ULC_appr}
\end{aligned}
\end{align}

For the numerical integration, this study used the approximation formula of Eq.~\eqref{Eq:ULC_appr} with $T'=T- 0.004$ (i.e., up to one day before maturity).
In the same way as the Monte Carlo simulation, the delta hedge process stops at $T'$.
Table~\ref{Table:comparison} lists the numerical integration of $I$ with various $T$ and $K$, and the corresponding simulation results are reported in the parentheses.
For the simulation and numerical integration, $\sigma=0.3, r=0.05$ ($S_0 =1$ and $\alpha =1$ are assumed implicitly).
A time grid $t = [0, T’]$ with interval size $\Delta t = 0.0001$, and a strike price grid $k = [0.5,1.5]$ with interval size $\Delta k = 0.001$, are constructed in the numerical integration.
For $\phi$, the Black-Scholes option prices are used.

In the first simulation, the delta hedging portfolio is rebalanced every hour, and the number of total sample paths is $10^{4}$; 
this procedure mimics continuous trading.
In the second simulation, a threshold is fixed to $\Delta D = 0.05$, and 
the portfolio is rebalanced whenever the change in the delta process exceeds the threshold.
Therefore, this procedure is for discrete trading.
The values reported in the table show that the numerical integrations of $I$ are similar to the means of corresponding simulation results within the entire range of $K$ and $T$.
The table also shows the 99\% confidence intervals and almost all confidence intervals calculated by simulation results contain the numerical integration values.

Second, the properties of the unit liquidity cost are examined by plotting various surfaces of $I$ with respect to the volatility, maturity, and moneyness.
Figures~\ref{Fig:ELC_sigma_T_K09}, \ref{Fig:ELC_sigma_T} and \ref{Fig:ELC_sigma_T_K11} plot the surface of $I$ as a function of $\sigma$ and $T$ with fixed $K=0.9, 1.0$, and $1.1$, respectively.
If a European call option is at-the-money (i.e., $K=1$), the unit expected liquidity cost, $I$, does not vary significantly throughout $\sigma$ and $T$, and values around 0.21 are obtained.
If the volatility and maturity are relatively small, then the expected liquidity cost increases 
with increasing volatility and maturity, which is consistent with our expectation.

In contrast to the previous case, if the volatility and maturity are relatively large, then the expected liquidity cost decreases as volatility and maturity increase.
The delta of a European option varies mostly around at-the-money, similar to the expected liquidity cost.
An at-the-money European option with a larger maturity and volatility has a higher probability of ending up out-of-the-money or in-the-money than an at-the-money European option with near maturity and low volatility.
Therefore, an at-the-money European option with large maturity and volatility might have a lower liquidity cost than an at-the-money European option with smaller maturity and less volatility.

In addition, although the European call option is not at-the-money (i.e., $K=0.9$ or $1.1$, if $\sigma$ is sufficiently large), then the unit expected liquidity cost, $I$, does not vary much, in the same way as maturity $T$.
As it is a European call option, delta is more sensitive closer to its expiration, and
the delta hedging liquidity cost is concentrated near the maturity (i.e., the liquidity cost when it is far from maturity is relatively insignificant).

Figure~\ref{Fig:ELC_K_T} plots the surface of $I$ as a function of $K$ and $T$ with fixed $\sigma = 0.3$.
The expected liquidity cost increases with increasing maturity for all values over $K$.
On the other hand, if a European option is at-the-money, the expected liquidity cost increases slowly compared to the cost of the in-the-money or out-of-the-money European options.
For in-the-money and out-of-the-money options, the expected liquidity cost increases significantly as the maturity $T$ increases.
In addition, the expected liquidity cost of the at-the-money option is generally larger than the cost of the in-the-money or out-of-the-money options.
The reason is that the moneyness of the deep in-the-money and out-of-the-money European options with a short time to maturity have little probability of changing;
therefore, the deltas of the options do not change significantly, as a small change in delta implies a low liquidity cost.

Figure~\ref{Fig:ELC_sigma_K} plots the surface of $I$ as a function of $\sigma$ and $K$ with fixed maturity $T=0.1$.
Similar to the previous cases, if a European option is at-the-money, the expected liquidity cost is not significantly affected by $\sigma$.
For an at-the-money European option, delta is sufficiently sensitive to the movement of the underlying price, regardless of the volatility level.
As a result, the liquidity costs of at-the-money European options are large for all $\sigma$.
For in-the-money and out-of-the-money options, the liquidity cost increases with increasing $\sigma$.
When there is more volatility, the moneyness is more likely to be altered from in-the-money to out-of-the-money, or vice versa.
In addition, the expected liquidity cost of the at-the-money option is larger than the cost of the in-the-money or out-of-the-money options.

\begin{table}
	\centering
	\caption{Comparison of the numerical integration $I$ and the corresponding simulation test with various $K$ and $T$ with $\sigma=0.3, r=0.05$ with mean values of simulations and the half-length of 99\% confidence interval (C.I.)}\label{Table:comparison}
	\begin{tabular}{ccccccc}
		\hline
		& & & & $K$\\
		\hline
		method & $T$ & 0.8 & 0.9 & 1.0 & 1.1 & 1.2 \\
		\hline 
		numerical & 0.1 & 0.0049 & 0.0791 & 0.2040 & 0.1195 & 0.0245\\
		simul. 1 &  & 0.0046 & 0.0775 & 0.2039 & 0.1213 & 0.0243\\
		C.I. & & $\pm 0.0006$ & $\pm 0.0024$ & $\pm 0.0030$ & $\pm 0.0033$ & $\pm 0.0017$\\
		simul. 2 &  & 0.0045 & 0.0772 & 0.2035 & 0.1207 & 0.0231\\
		C.I. & & $\pm 0.0006$ & $\pm 0.0024$ & $\pm 0.0030$ & $\pm 0.0033$ & $\pm 0.0017$ \\
		\hline
		numerical & 0.2 & 0.0283 & 0.1259 & 0.2165 & 0.1774 & 0.0844\\
		simul. 1 & & 0.0275 & 0.1268 & 0.2168 & 0.1793 & 0.0838\\
		C.I. & & $\pm 0.0016$ & $\pm 0.0032$ & $\pm 0.0035$ & $\pm 0.0039$ & $0.0033$ \\
		simul. 2 & & 0.0263 & 0.1260 & 0.2161 & 0.1787 & 0.0835 \\
		C.I. & & $\pm 0.0016$ & $\pm 0.0031$ & $\pm 0.0035$ & $\pm 0.0039$ & $0.0033$\\
		\hline
		numerical & 0.5 & 0.0834 & 0.1672 & 0.2259 & 0.2256 & 0.1797\\
		simul. 1 & & 0.0833 & 0.1674 & 0.2279 & 0.2228 & 0.1793  \\
		C.I. & & $\pm 0.0028$ & $\pm 0.0036$ & $\pm 0.0040$ & $ \pm 0.0043$ & $\pm 0.0046$ \\
		simul. 2 & & 0.0826 & 0.1669 & 0.2272 & 0.2223 & 0.1784 \\
		C.I. & & $\pm 0.0028$ & $\pm 0.0036$ & $\pm 0.0040$ & $ \pm 0.0043$ & $\pm 0.0046$ \\
		\hline
		numerical & 1.0 & 0.1194 & 0.1829 & 0.2297 & 0.2428 & 0.2293\\
		simul. 1 & & 0.1215 & 0.1821 & 0.2290 & 0.2423 & 0.2303 \\
		C.I. & & $\pm 0.0033$ & $\pm 0.0038$ & $\pm 0.0043$ & $\pm 0.0046$ & $\pm 0.0051$\\
		simul. 2 & & 0.1203 & 0.1811 & 0.2289 & 0.2412 & 0.2293 \\
		C.I. & & $\pm 0.0033$ & $\pm 0.0038$ & $\pm 0.0043$ & $\pm 0.0046$ & $\pm 0.0050$\\
		\hline
	\end{tabular}
\end{table}


\begin{figure}
\centering
\includegraphics[width=0.6\textwidth]{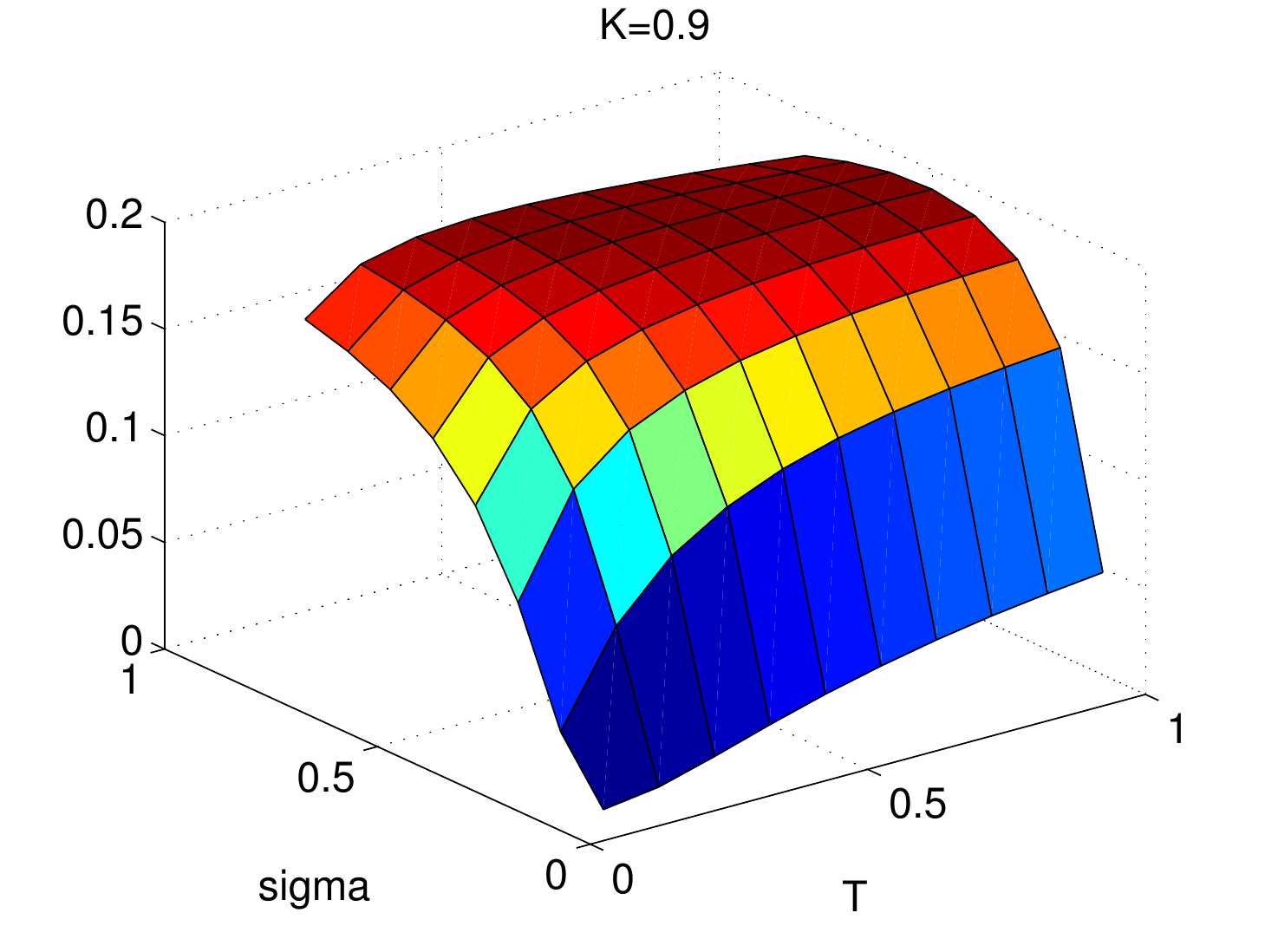}
\caption{Unit expected liquidity cost $I$ for various $\sigma$ and $T$ with $K=0.9$}\label{Fig:ELC_sigma_T_K09}
\end{figure}

\begin{figure}
\centering
\includegraphics[width=0.6\textwidth]{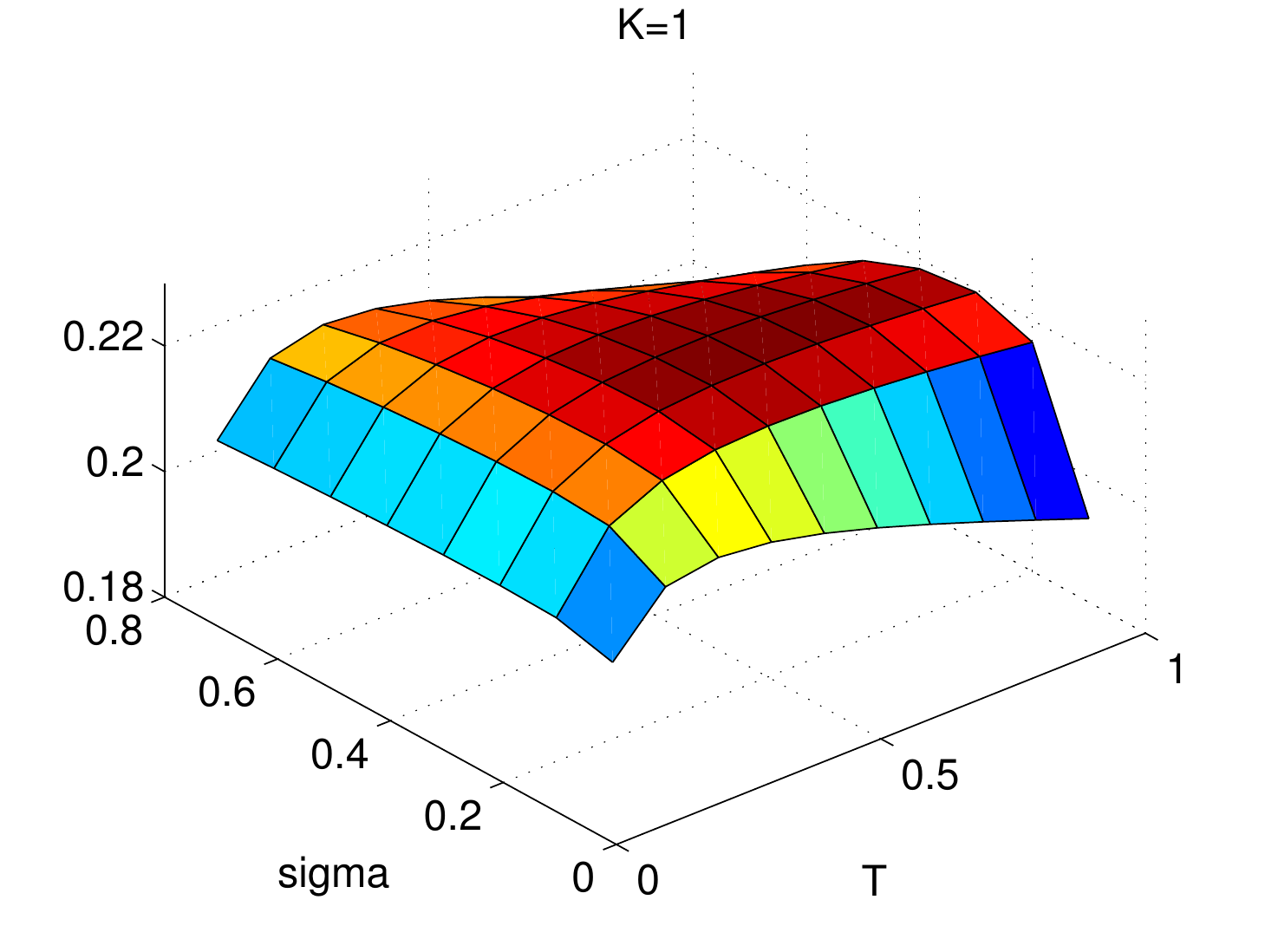}
\caption{Unit expected liquidity cost $I$ for various $\sigma$ and $T$ with $K=1$}\label{Fig:ELC_sigma_T}
\end{figure}

\begin{figure}
\centering
\includegraphics[width=0.6\textwidth]{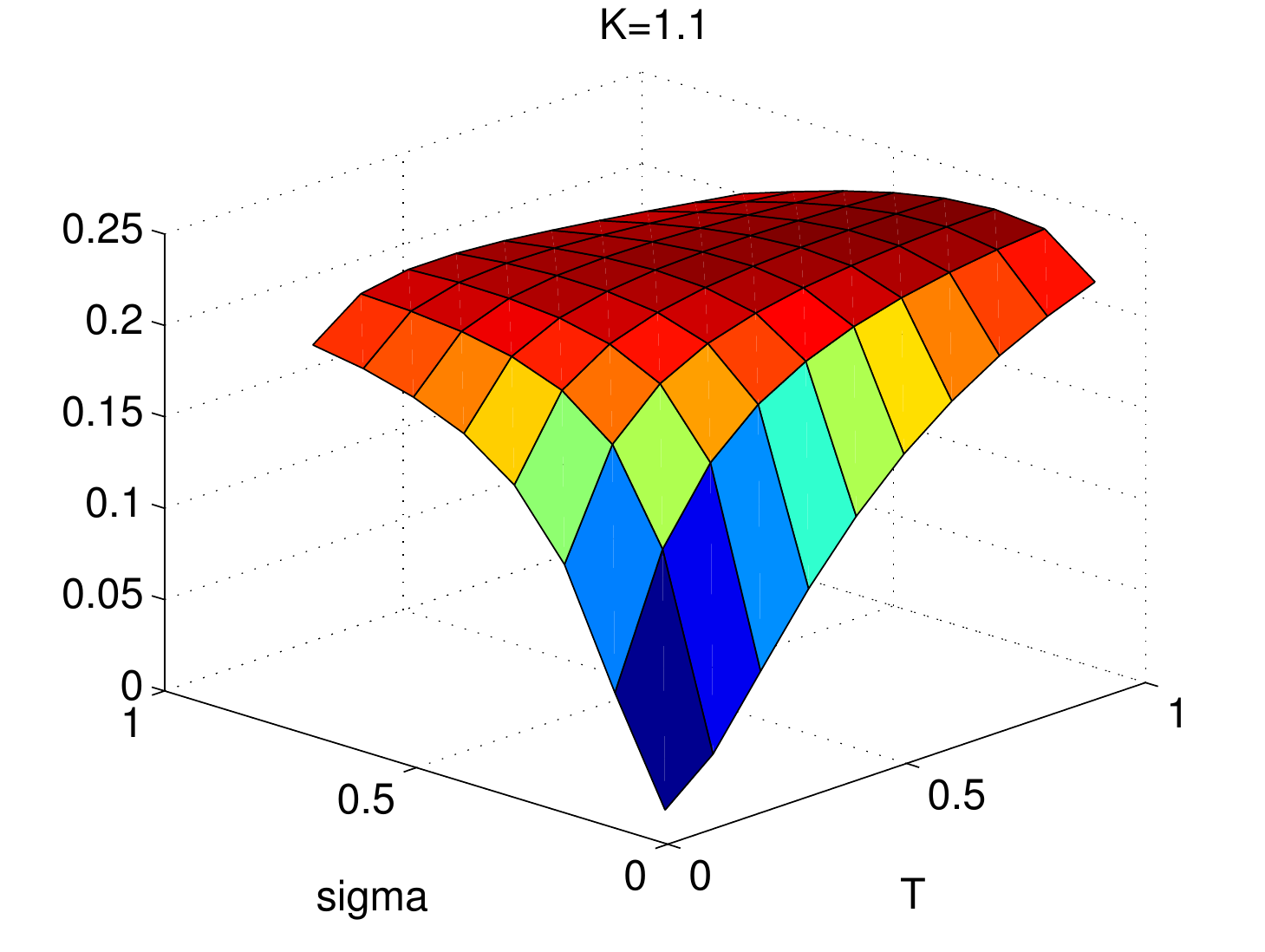}
\caption{Unit expected liquidity cost $I$ for various $\sigma$ and $T$ with $K=1.1$}\label{Fig:ELC_sigma_T_K11}
\end{figure}

\begin{figure}
\centering
\includegraphics[width=0.6\textwidth]{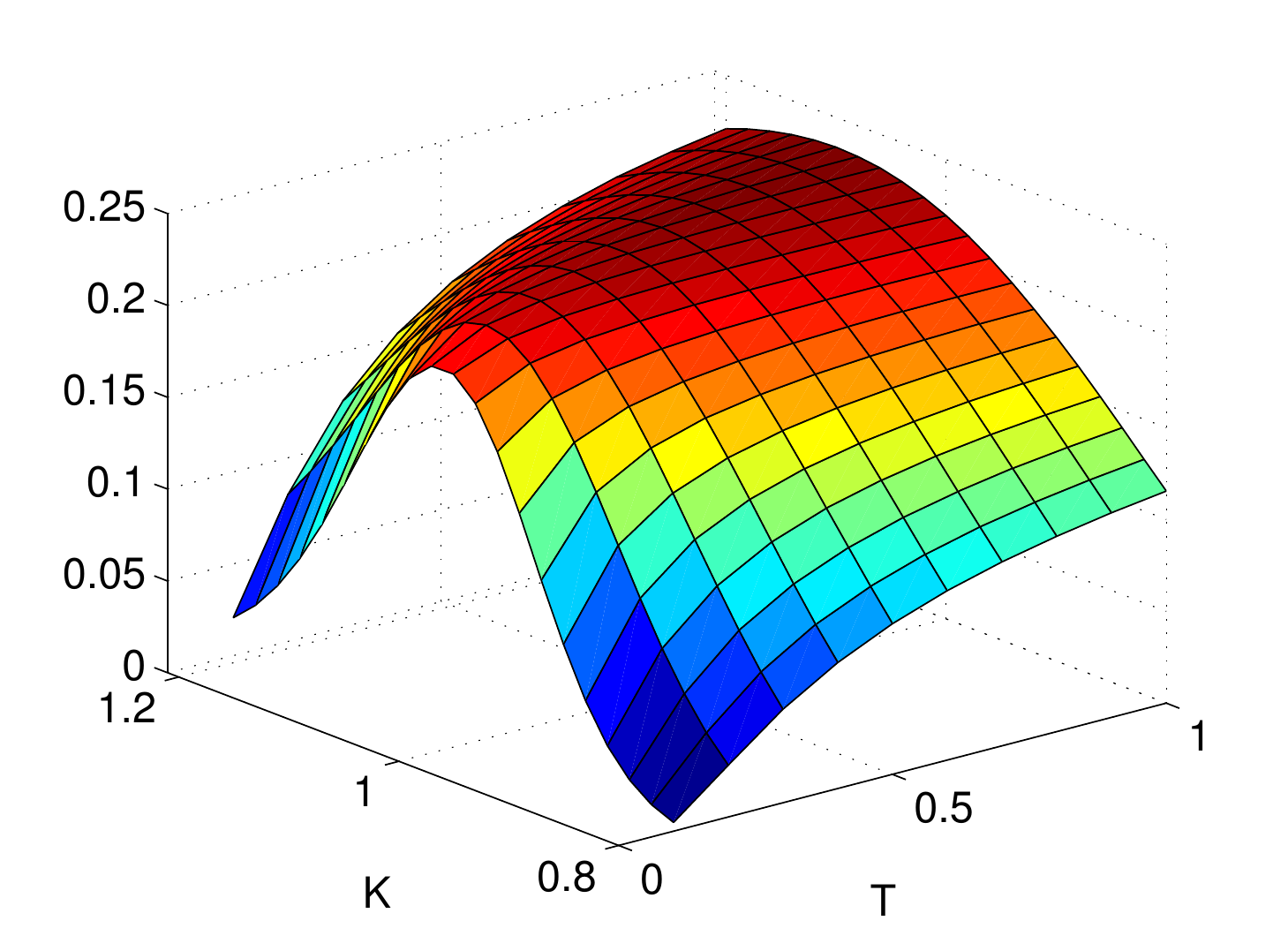}
\caption{Unit expected liquidity cost $I$ for various $K$ and $T$ with $\sigma=0.3$}\label{Fig:ELC_K_T}
\end{figure}

\begin{figure}
\centering
\includegraphics[width=0.6\textwidth]{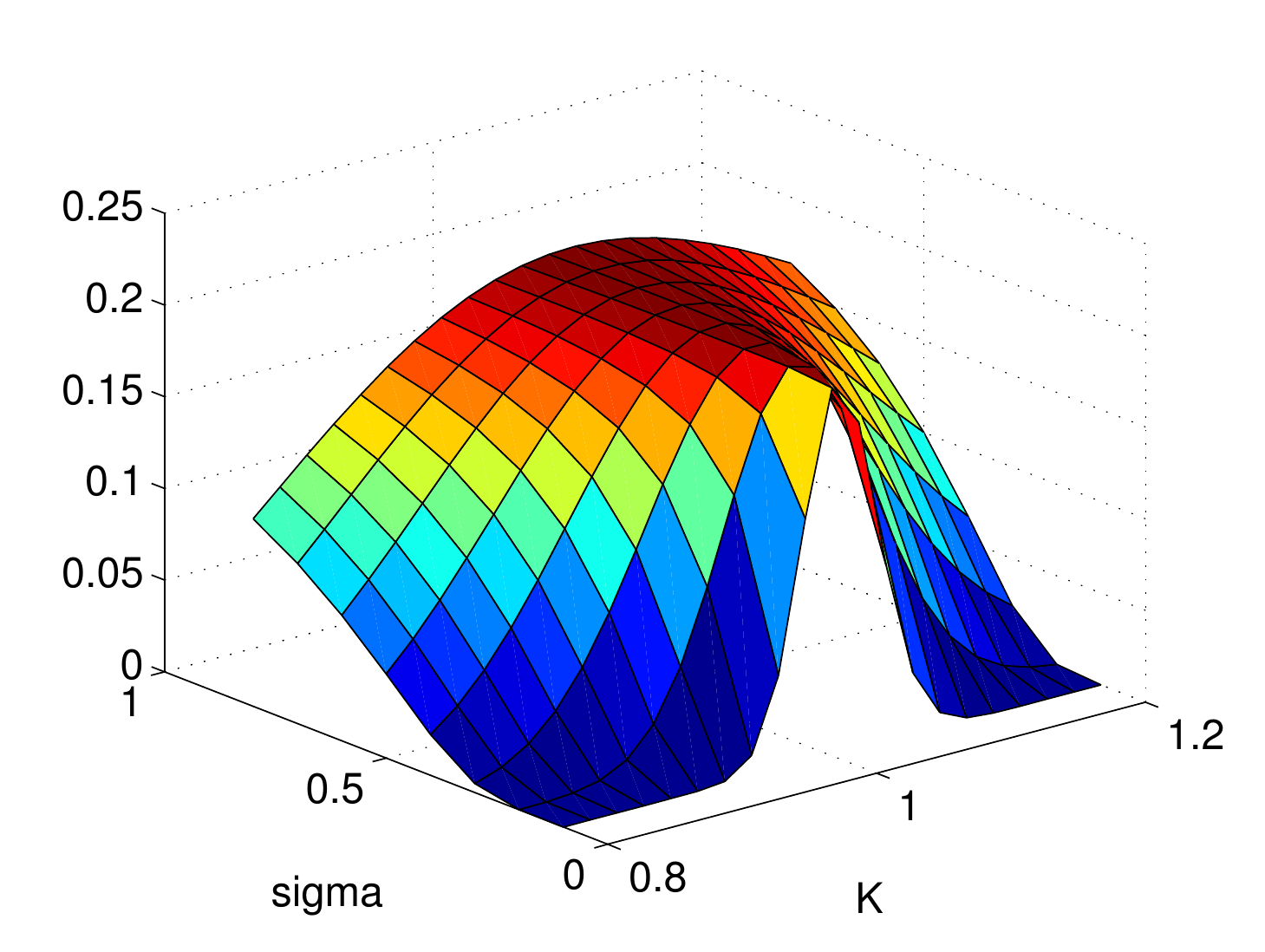}
\caption{Unit expected liquidity cost $I$ for various $\sigma$ and $K$ with $T=0.1$}\label{Fig:ELC_sigma_K}
\end{figure}

\section{Distribution of liquidity cost}\label{Sect:dist}
This section examines the distribution of the liquidity cost in the special case of discrete trading and linear supply curve.
Delta hedging is performed with a fixed time interval, and this setting is different from Section~\ref{Sect:disc}.
The computation of the distribution is based on a recursive formula for a numerical integration, 
and it takes a longer time than the computation of an expected liquidity cost, but a shorter time than a Monte Carlo simulation.

Recall that the liquidity cost of a European option with strike $K$, under a discrete delta hedging and the above condition, is
\begin{equation}
L = \sum_{i=1}^{N} \alpha' S_{i-1} \left(\frac{\partial D_{N-i+1}(S_{i-1}, K)}{\partial x}\right)^2(S_i - S_{i-1})^2 \label{Eq:LC_disc2}\end{equation}
which is similar to Eq.~\eqref{Eq:LC_disc} but over an equidistant partition $0=t_0 < t_1 < \cdots < t_N = T$
and the notation of the delta function is slightly different from the previous sections.
The subscript of $D$ denotes the time to maturity, and $S_{i-1}$ is at the time $t_{i-1}$ stock price.
The partial derivative $\partial x$ implies the partial derivative with respect to the first argument of $D$, the stock price (i.e., the Gamma of the option price).
Calculating the distribution of $L$ is equivalent to calculating the expectation of $(L-\xi)^+$ for some constant $\xi$.
The calculation of the expectation of $(L-\xi)^+$ is similar to a derivation of an Asian-type call option, 
and the distribution of $L$ is derived by differentiating the call option price twice with respect to $\xi$ \citep{Breeden1978}.
To calculate the distribution of $L$, a recursive method, which was introduced in \cite{Lee}, \cite{CLP}, and \cite{Park} is extended and applied.

The numerical method to compute the distribution of the liquidity cost is explained briefly.
Let $g_{n}^{N}(x_0, \cdots, x_n, K, \xi)$ represent the $\F_n$-conditional expectation of $L$ such that
\begin{align}
\begin{aligned}
& g_{n}^{N}(S_0, \cdots, S_n, K, \xi) = \E^{\Q}[(L-\xi)^+ | \F_n ] \\
&=\E^{\Q} \left[\left. \left(\sum_{i=1}^{N} \alpha' S_{i-1} \left(\frac{\partial D_{N-i+1}(S_{i-1},K)}{\partial x}\right)^2(S_i - S_{i-1})^2  - \xi \right)^+  \right| \F_n \right].
\end{aligned}
\end{align}
The superscript, $N$, of $g$ denotes the total number of terms inside the summation, and the subscript
$(n)$ denotes the time when the expectation is performed.
As in the previous section, the risk-neutral probability is used, but the choice of the measure depends simply on the practitioner's purpose.

The goal is the time 0 expectation, for example,
\begin{equation}
g_0^N(S_0, K, \xi) = \E^{\Q}[(L-\xi)^+].
\end{equation}
Since
\begin{equation}
g_n^N(S_0, \cdots, S_n, K, \xi) = \E^{\Q}[ g_{n+1}^N (S_0, \cdots, S_n, S_{n+1}, K, \xi) | \F_n ],
\end{equation}
the following integration form can be derived:
\begin{equation}
g_n^N(x_0, \cdots, x_n, K, \xi) = \int_0^\infty g_{n+1}^N (x_0, \cdots, x_n, x_{n+1}, K, \xi)p(x_{n+1};x_n) \D x_{n+1}\label{Eq:recursive1}
\end{equation}
where $p$ is the transition function for the stock price from $x_n$ to $x_{n+1}$ over the time interval $[t_n, t_{n+1}]$ under the corresponding measure.
By rearranging the strike price of the `Asian-type option', 
\begin{equation}
g_n^N(x_0, \cdots, x_n, K, \xi) = g_0^{N-n}(x_n, K, \xi_n) \label{Eq:relation1}
\end{equation}
where
\begin{equation}
\xi_n = \xi - \sum_{i=1}^{n} \ell_i := \xi - \sum_{i=1}^{n} \alpha' x_{i-1} \left(\frac{\partial D_{N-i+1}(x_{i-1},K)}{\partial x}\right)^2(x_i - x_{i-1})^2.
\end{equation}

By denominating the Asian-type option price by $x_n$,
\begin{equation}
g_0^{N-n}(x_n, K, \xi_n) = x_n g_0^{N-n}\left(1, \frac{K}{x_n}, \frac{\xi_n}{x_n} \right),
\end{equation}
and by Eq.~\eqref{Eq:recursive1},
\begin{equation}
x_n g_0^{N-n}\left(1, \frac{K}{x_n}, \frac{\xi_n}{x_n} \right) = \int_0^\infty x_{n+1} g_0^{N-n-1}\left(1, \frac{K}{x_{n+1}}, \frac{\xi_{n+1}}{x_{n+1}} \right)p(x_{n+1};x_n) \D x_{n+1}.
\end{equation}
Since the above equation holds for all $x_n>0$, by putting $x_n=1$, and using the relation between $\xi_n$ and $\xi_{n+1}$,
\begin{equation}
g_0^{N-n}\left(1, K, \xi_n \right) = \int_0^\infty x_{n+1} g_0^{N-n-1}\left(1, \frac{K}{x_{n+1}}, \frac{\xi_{n} - \ell_{n+1}}{x_{n+1}} \right)p(x_{n+1}; 1) \D x_{n+1}.\label{Eq:recursive2}
\end{equation}
In particular, when $n=N-1$, by Eq.~\eqref{Eq:relation1} and the definition of $g$,
\begin{align}
\begin{aligned}
g_0^1(1, K, \xi_{N-1}) &= g_{N-1}^{N}(x_0,\cdots,x_{N-1},K,\xi)\\
&=\int_0^\infty \left( \alpha' \left(\frac{\partial D_1(1,K)}{\partial x}\right)^2(x_N - 1)^2  - \xi_{N-1} \right)^+ p(x_N;1) \D x_N.
\end{aligned}
\end{align}
Repeating the numerical integration from $g_0^1$ to $g_0^N$ and applying
$ g_0^N(x, K, \xi) = x g_0^N(1, K/x, \xi/x)$
the Asian-type option price, the distribution of the liquidity cost can be calculated.

In Figure~\ref{Fig:distribution1}, ~\ref{Fig:distribution2}, ~\ref{Fig:distribution3} the comparison between the distribution calculated using the numerical method (solid line) and the Monte Carlo simulation (bar graph) is presented.
The numerical results are quite close to the Monte Carlo simulation results over all $\sigma$, strike $K$ and maturity $T$.

\begin{figure}
	\centering
		\includegraphics[width=0.32\textwidth]{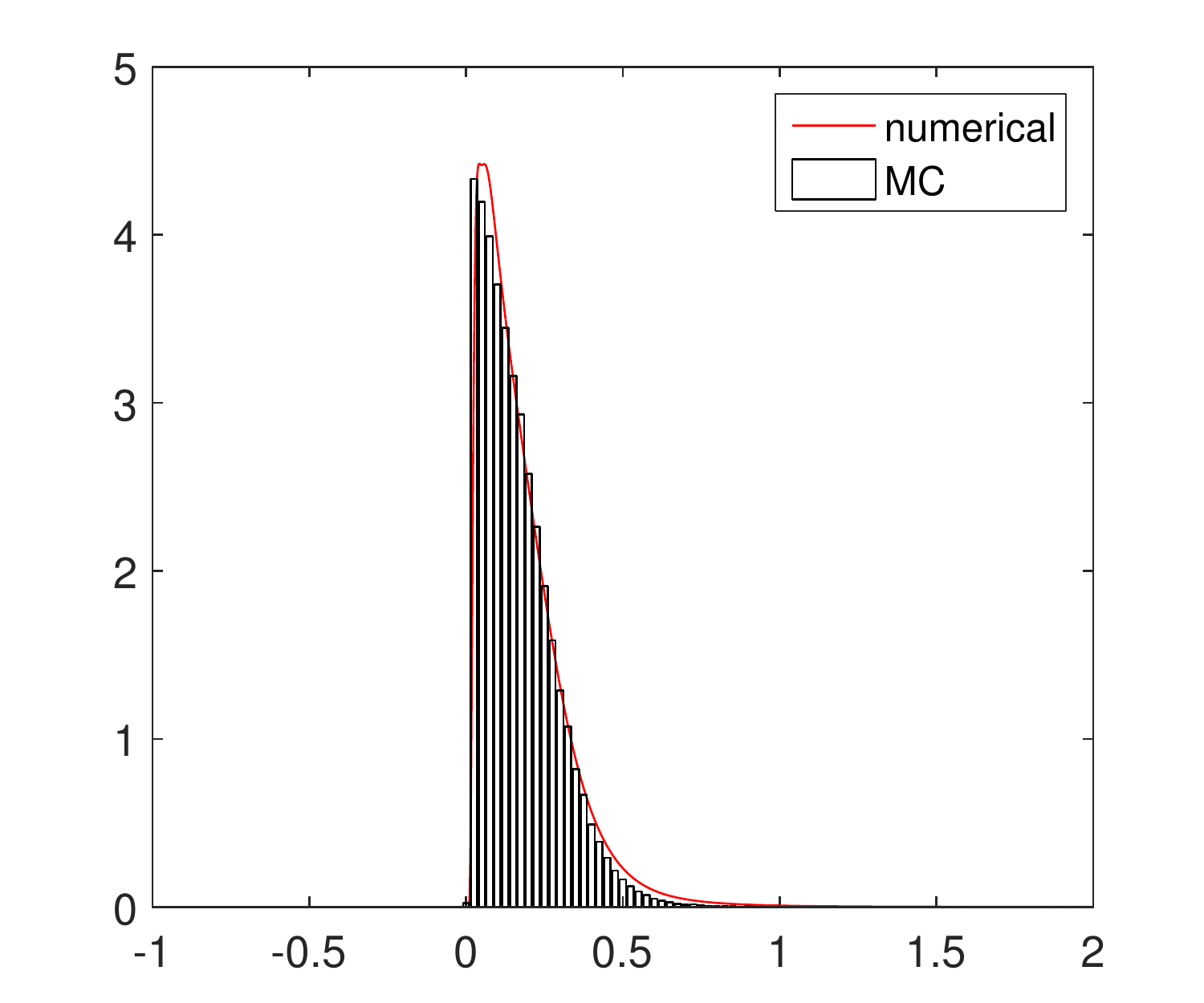}
		\includegraphics[width=0.32\textwidth]{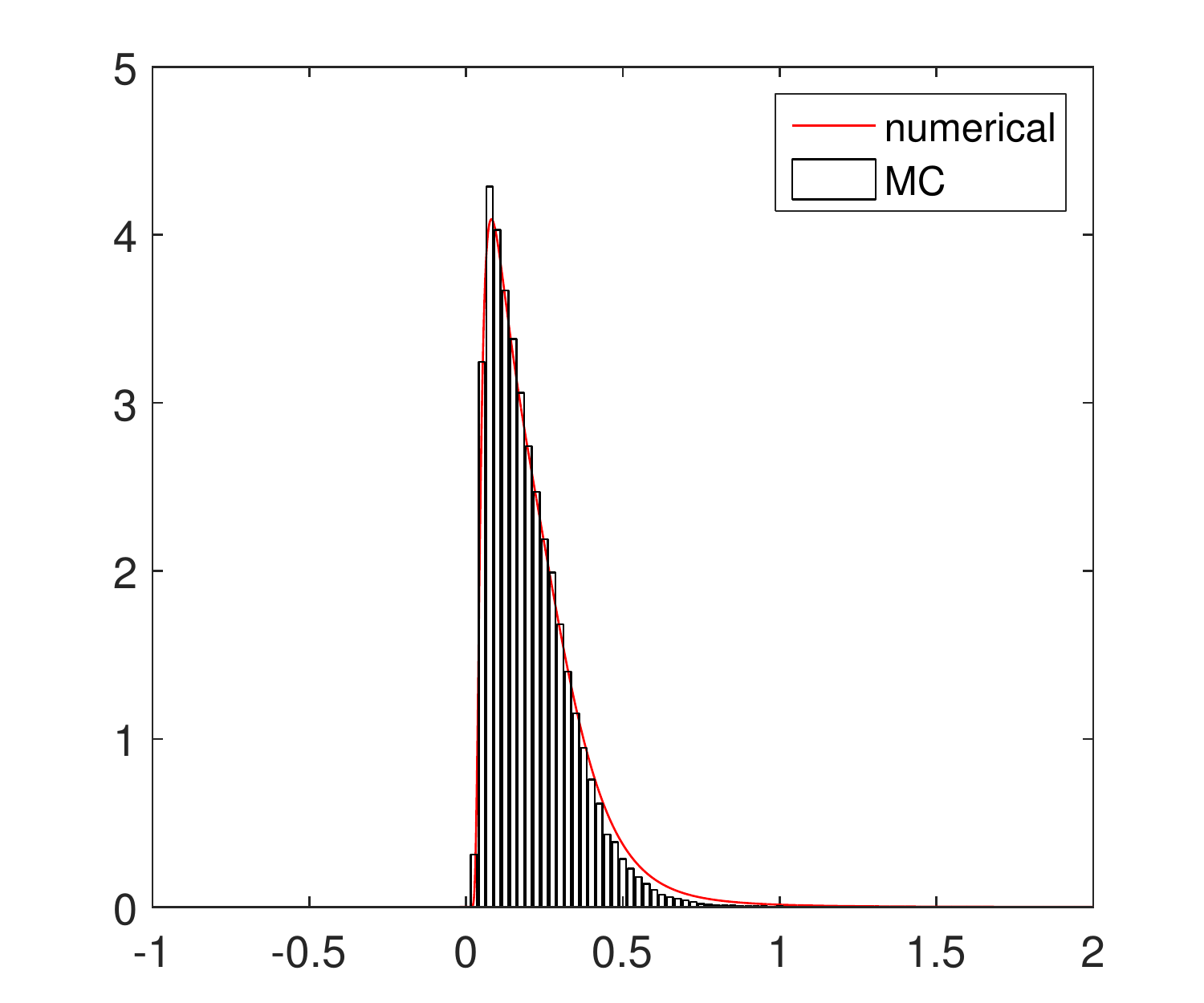}
		\includegraphics[width=0.32\textwidth]{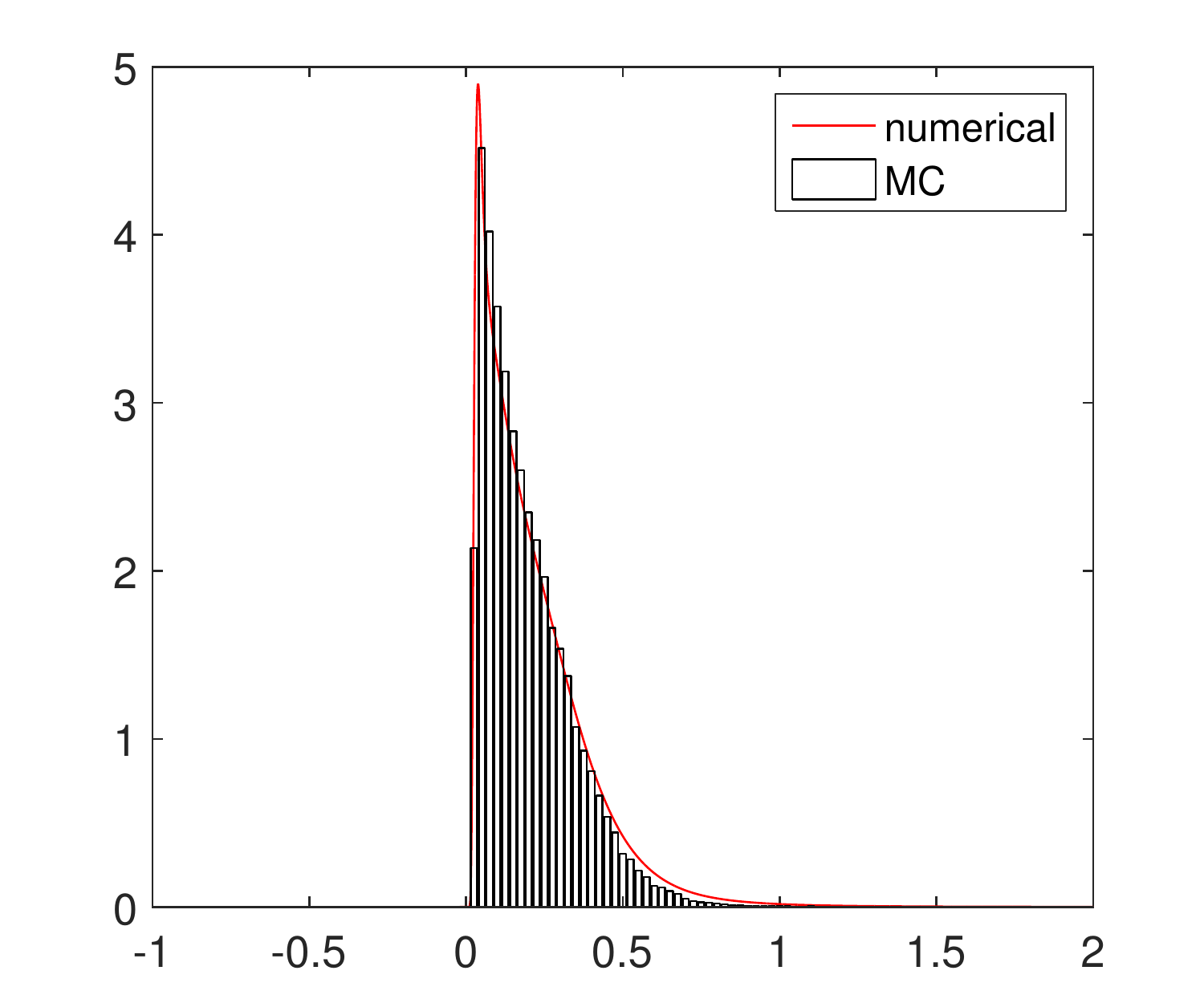}
	\caption{Distribution of liquidity errors with $\sigma=0.8$, $T=0.1$ and $K=0.9$ (left) $ K=1$ (center) $K=1.1$ (right)}\label{Fig:distribution1}
\end{figure}

\begin{figure}
	\centering
		\includegraphics[width=0.32\textwidth]{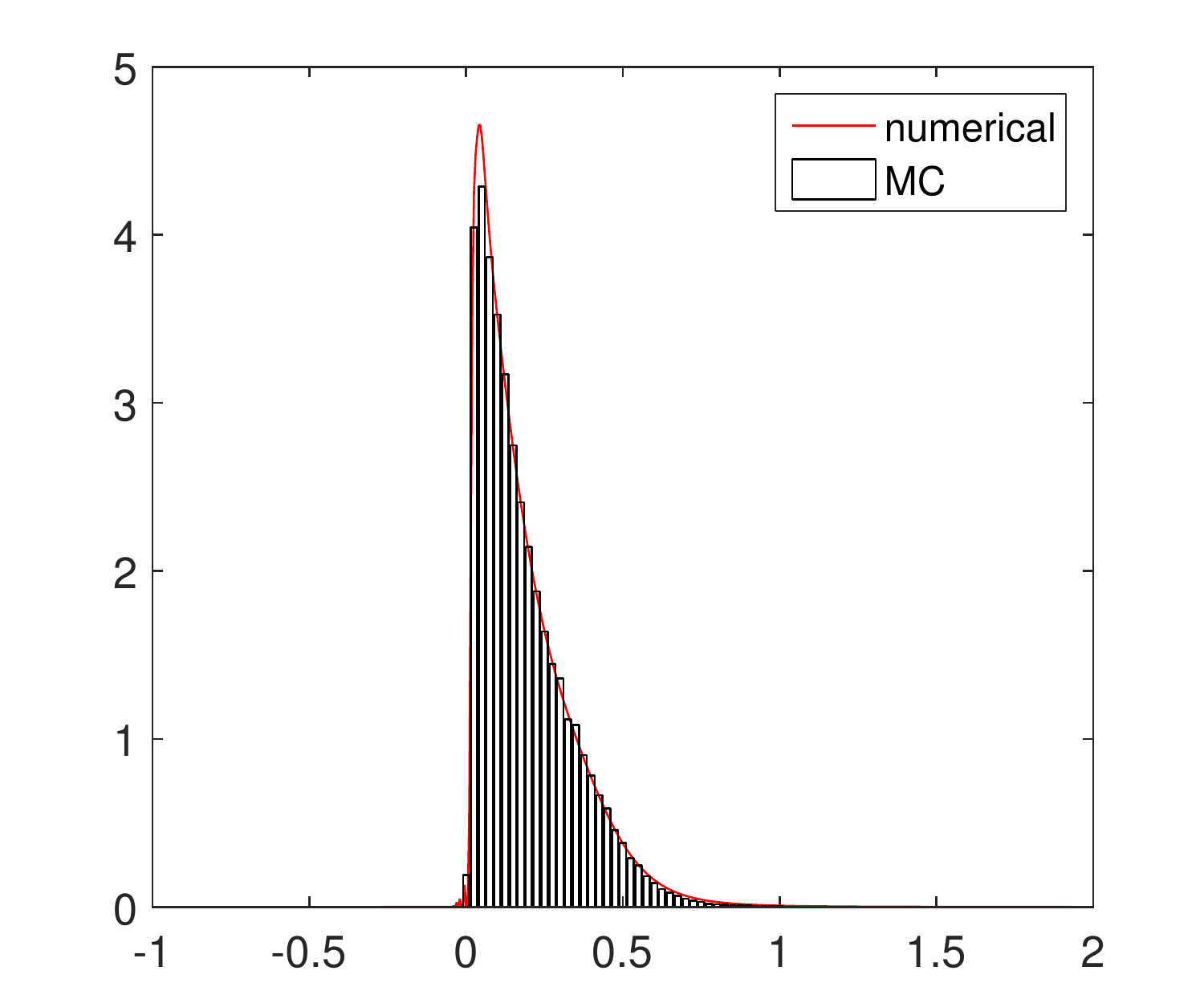}
		\includegraphics[width=0.32\textwidth]{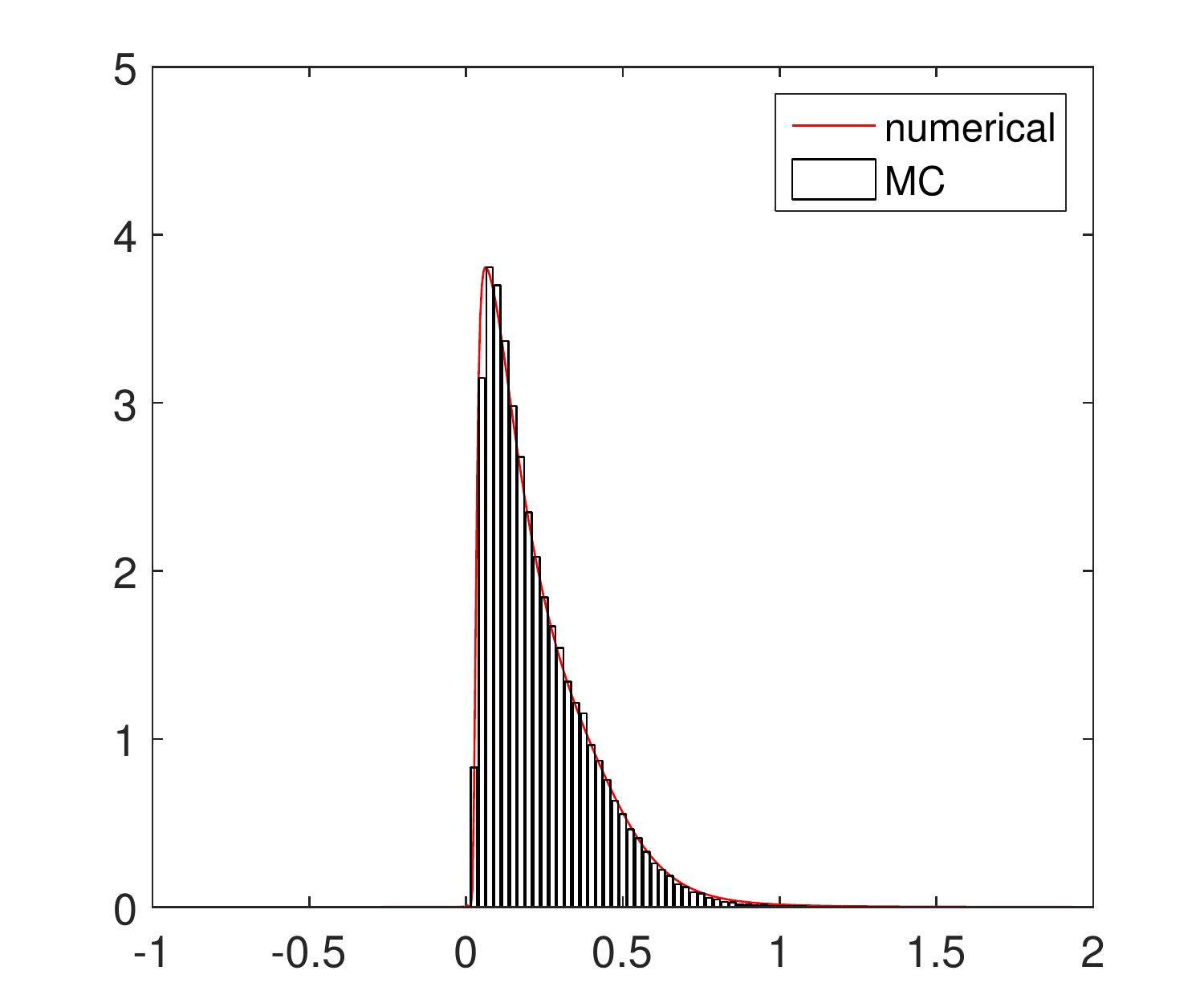}
		\includegraphics[width=0.32\textwidth]{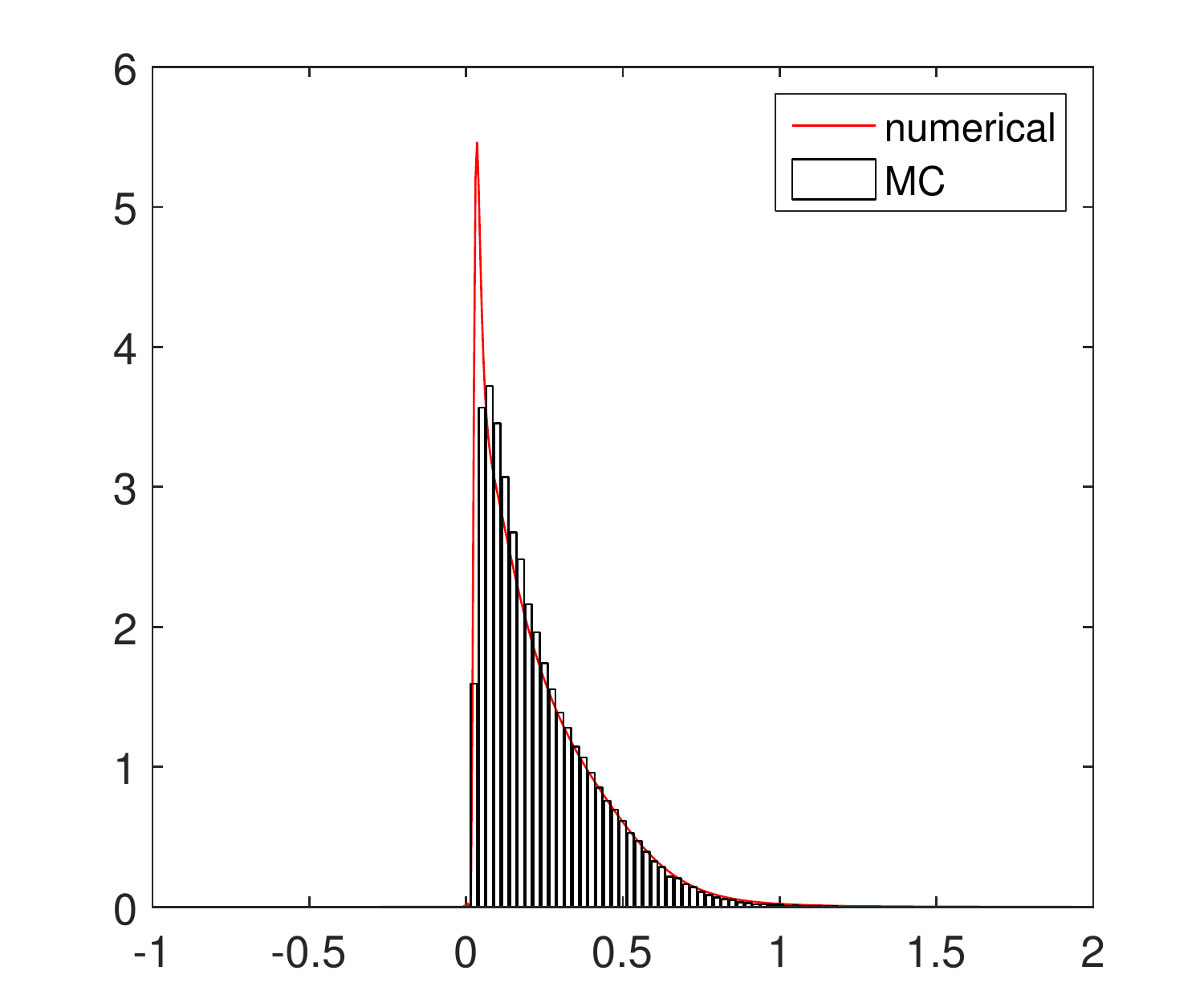}
	\caption{Distribution of liquidity errors with $\sigma=0.4$, $T=0.5$ and $K=0.9$ (left) $ K=1$ (center) $K=1.1$ (right)}\label{Fig:distribution2}
\end{figure}

\begin{figure}
	\centering
		\includegraphics[width=0.32\textwidth]{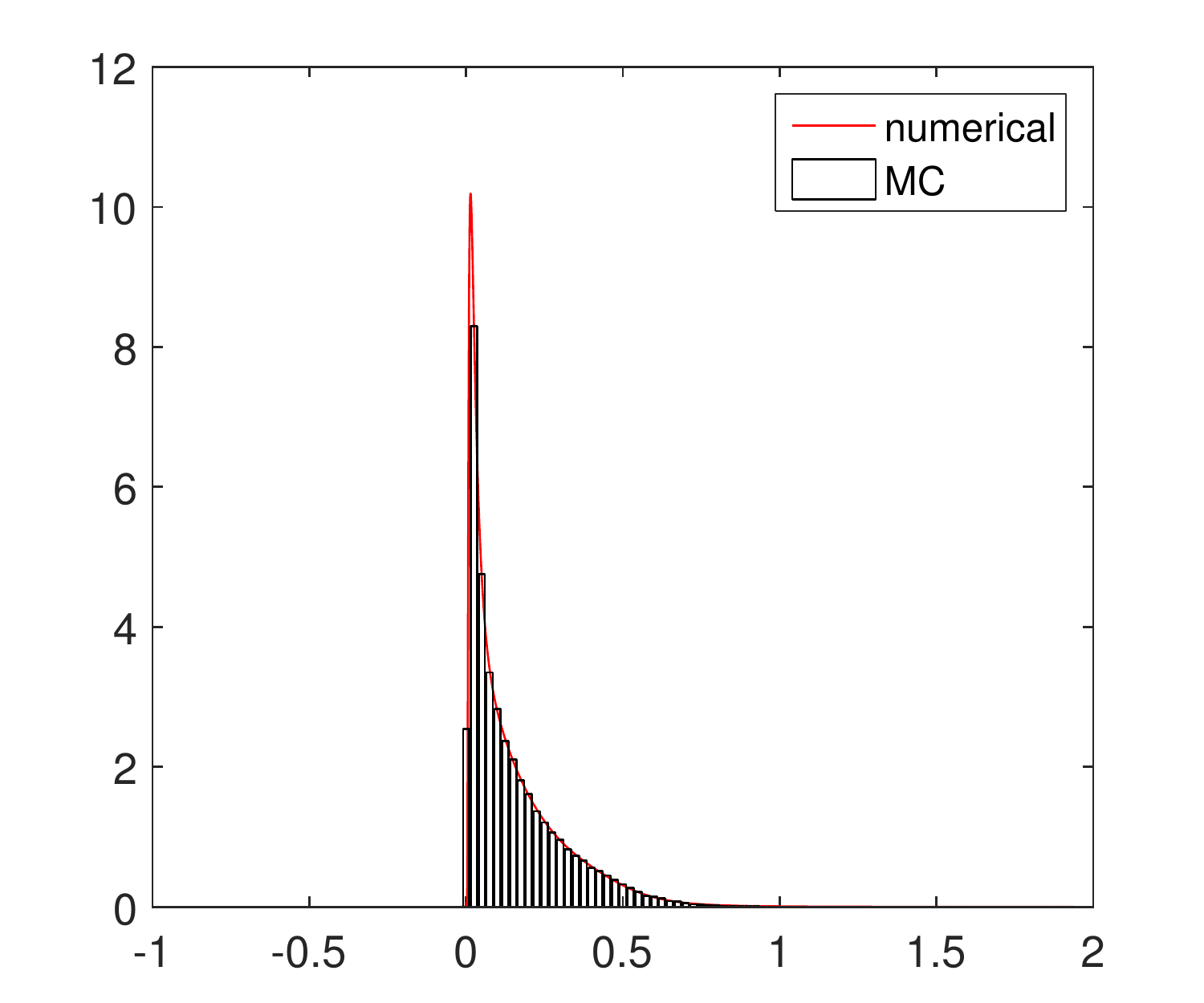}
		\includegraphics[width=0.32\textwidth]{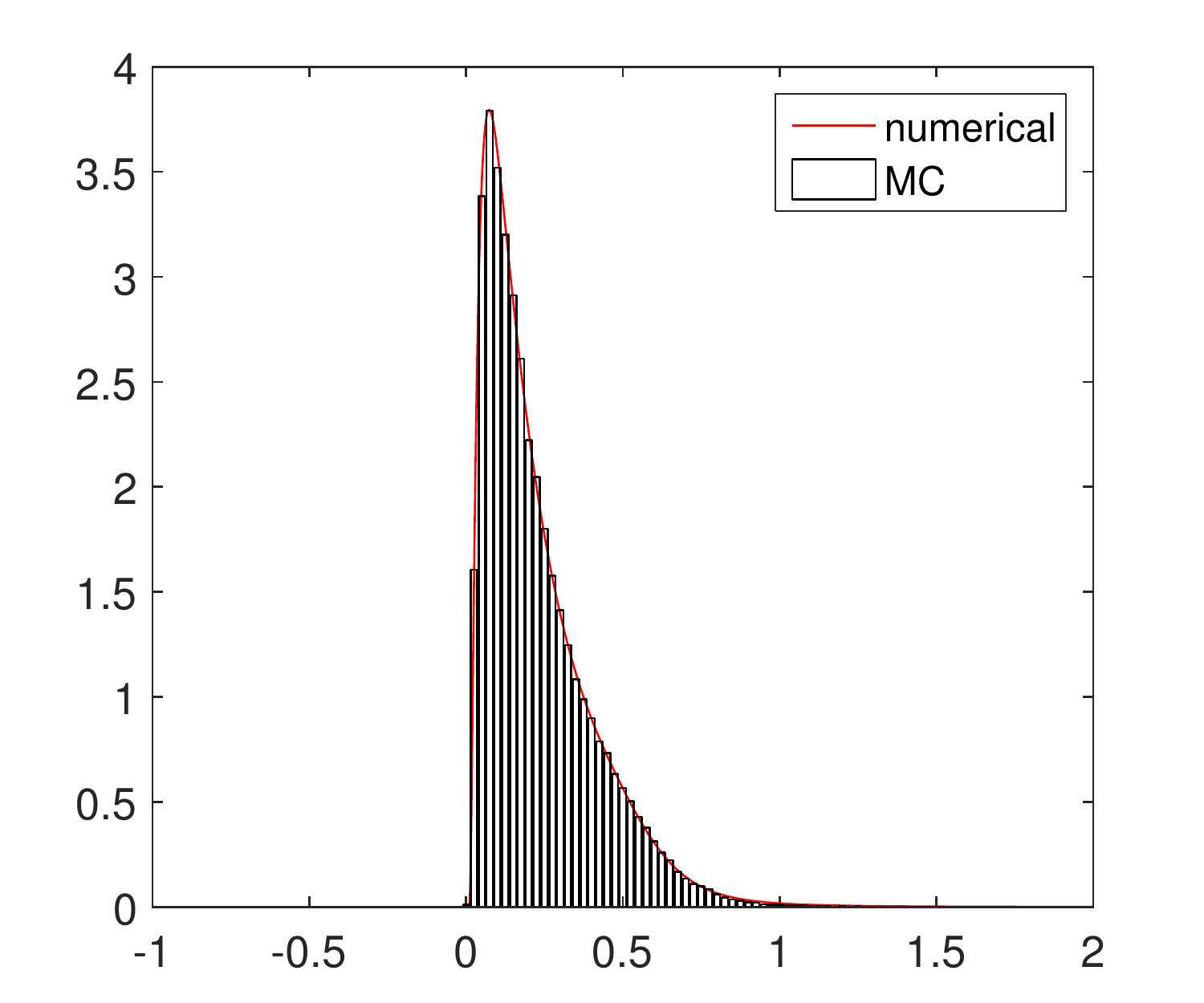}
		\includegraphics[width=0.32\textwidth]{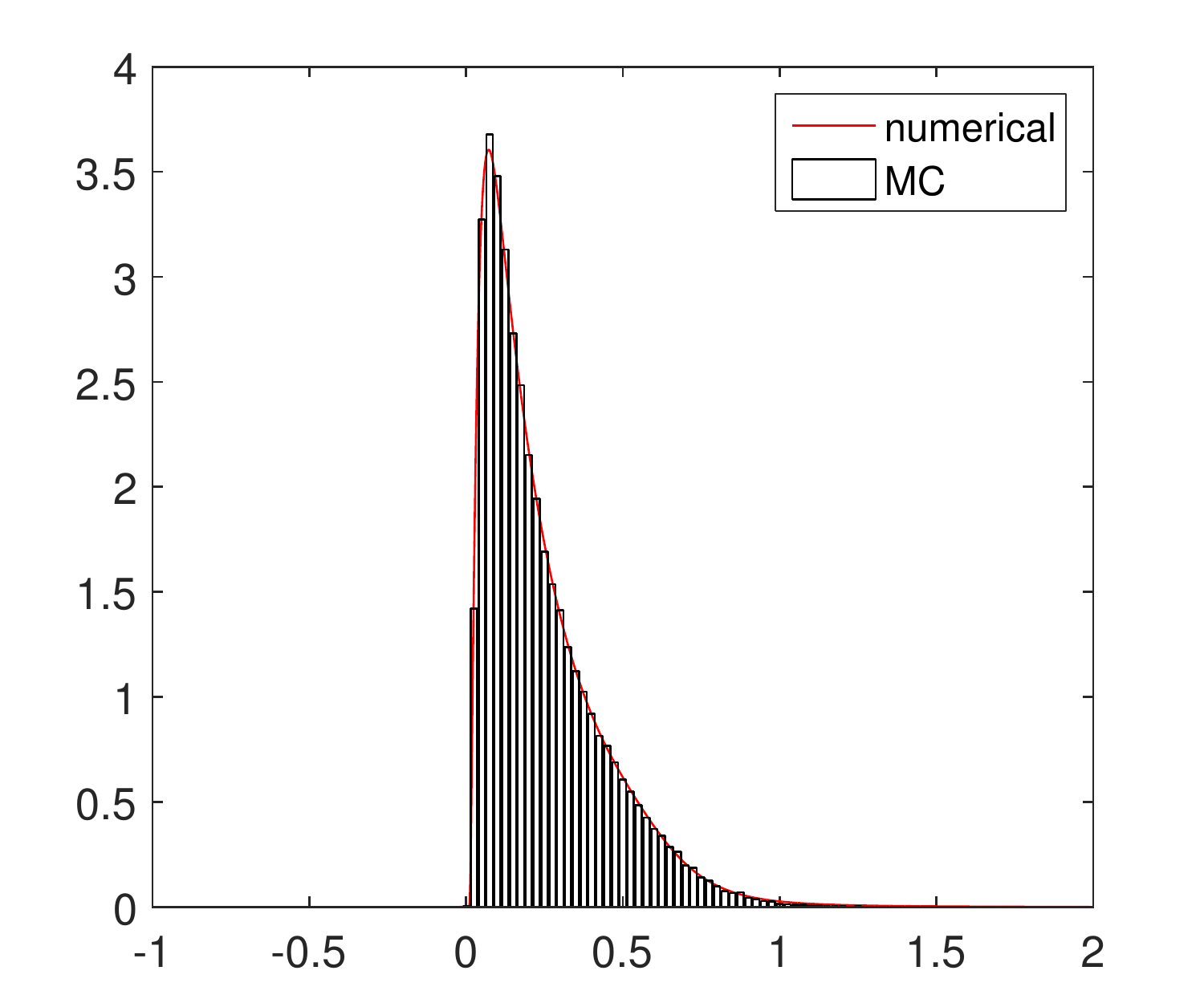}
	\caption{Distribution of liquidity errors with $\sigma=0.2$, $T=1$ and $K=0.9$ (left) $ K=1$ (center) $K=1.1$ (right)}\label{Fig:distribution3}
\end{figure}

\section{Supply curve and order book model}\label{Sect:LOB}
The expected liquidity cost depends on $\alpha$, and in practice, $\alpha$ depends on the estimation procedure of the supply curve.
If the supply curve is discontinuous, then continuous hedging is excluded, and discrete trading should be considered.

This section discusses how to apply the existing order book model to our method for calculating the liquidity cost.
If the limit order model is assumed to be stochastic, then the supply curve is also stochastic.
The dynamics of the supply curve, which in this approach are governed by parameter $\alpha_t$ or $\alpha'_t$, are independent from the stock price and the investor's hedging process.
The expected liquidity cost for a continuous hedging process is then represented by
\begin{equation}
\E \left[ \int_0^T \alpha_t S_t \D [D]_t \right] = \int_0^T \E [\alpha_t] \E \left[\sigma^2_t S^3_t \left(\frac{\partial D(t,S_t)}{\partial x} \right)^2 \right] \D t
\end{equation}
Here, it is assumed that $\D [S]_t = \sigma^2_t S^2_t \D t$ for some volatility process and $\sigma_t$ for convenience.
A similar argument can be made for discrete trading.
The formula shows that only the expected slope of the supply curve $\E [\alpha_t]$ needs to be considered when calculating the total expected liquidity cost.
The relationship between the existing stochastic queuing model of the limit order book and the expected supply curve is explained.

\cite{Cont2010} proposed a stochastic model for order book dynamics where the limit and market orders and the cancellation of limit orders are described by independent Poisson processes.
In this framework, the limit orders are placed on a discrete price grid and the limit order arrivals depend on how far the price of the limit order is from the best quote (i.e., the relative price).
\cite{Biais1995}, \cite{Bouchaud2002}, and \cite{Potters2003} also showed that the order arrival rates depend on relative price rather than actual price.

The limit buy and sell orders arriving at a distance of $i$ ticks from the opposite best quote, or simply the $i$-th tick, are assumed to have the same intensity rate $\lambda(i)$.
The market buy and sell orders, which occur in only the best quotes of ask and bid, respectively, are assumed to arrive at independent and exponential times at a rate of $\mu$.
The rate of cancellation of the limit orders at a distance of the $i$-th tick from the opposite best quotes is proportional to the existing sizes of the limit orders.
If the total number of the outstanding limit orders is $y$, then the cancellation rate is $\theta(i)y$.
\cite{Cont2010} chose the size unit as the average size of the limit orders, but in this section, the size unit is assumed to be one without a loss of generality.

As the arrival rates of the bid and ask sides are assumed to be equal to each other, it is sufficient to consider only one side of bid and ask.
Let $N_t^{\ell(i)}$ be the total number of limit orders (of bid or ask side) at the $i$-th tick up to time $t$; let $N_t^{c(i)}$ be the total number of cancellations of limit orders at the $i$-th tick up to time $t$, and
$N_t^{m}$ the total number of market orders arrived up to $t$.
Under this assumption, the number of outstanding limit orders at time $t$ at the $i$-th tick is represented by three Poisson processes:
\begin{equation}
Y^{1}_t = N_t^{\ell(1)} - N_t^{m} - N_t^{c(1)} \quad \mathrm{and}\quad Y^{i}_t = N_t^{\ell(i)}  - N_t^{c(i)}, \quad \textrm{for } i>1. 
\end{equation}

For $i=1$, based on the above equation, the expected number of outstanding limit orders at time $t$, denoted by $y^{1}_t$, is represented by an integration equation
\begin{equation}
y^{1}_t = \lambda(1)t - \mu t - \theta(1) \int_0^t y^{1}_s \D s
\end{equation}
By solving the above equation, the expected outstanding limit orders is
\begin{equation}
y^{1}_t = \frac{\lambda(1) - \mu}{\theta(1)}\left(1 - \e^{-\theta(1) t}\right)
\end{equation}
where$y^{1}_0 = 0$.
As $t\rightarrow \infty$, 
\begin{equation}
y^{1}_t \rightarrow y^{1} := \frac{\lambda(1) - \mu}{\theta(1)}
\end{equation}
which is the expected number of outstanding limit orders at the first tick at the stationary state.
\begin{equation}
y^{i}_t \rightarrow y^{i} := \frac{\lambda(i) }{\theta(i)}.
\end{equation}

With these results, the expected stationary state limit order curve $m(q)$ and corresponding supply curve are defined as Eq.~\eqref{Eq:supply_curve}.
We have
\begin{equation}
m(q) = \mathrm{sign}(q) \sum_{i=1}^j y^{i}
\end{equation}
if $y^{j-1} < q \leq y^{j}$, for $j \geq 1$, and $y^0=0$.
Therefore, with the given exponential rates $\lambda(i), \theta(i)$, and $\mu$, one can derive the limit order and supply curves.
The rates can be estimated from the data consisting of the sequences of limit and market orders by counting the number of orders as described in \cite{Cont2010}.
In some studies, the limit order arrival function is fitted by a power law function of the form
\begin{equation}
\lambda(i) = \frac{k}{i^a}
\end{equation}
for some $a$ and $k$ \citep{Bouchaud2002,Zovko2002}.

Once the rates are estimated, the expected supply curve is derived; the remaining part of the calculation of the expected liquidity cost is straightforward, as described in Sections~\ref{Sect:conti}~and~\ref{Sect:disc}.

\section{Conclusion}\label{Sect:concl}
An integration formula was derived for the expected liquidity cost when performing the delta hedging process of a European option.
The liquidity cost is represented by the multiplication of the unit liquidity cost, current stock price, liquidity parameter, and the square of numbers of options being hedged.
The numerical procedure of the integration is correct and much faster than using the Monte Carlo method.
Interestingly, the expected liquidity cost of at-the-money European call options is significantly affected by volatility and maturity.
In general, the expected liquidity cost increases with increasing volatility and maturity.
This paper provided a numerical method for calculating the distribution of the liquidity cost.
Furthermore, the relationship between the queuing modeling of the order book and the present approach was explained.


\bibliography{Bib}
\bibliographystyle{elsarticle-harv}

\end{document}